\pdfoutput=1

\documentclass[
prl,
twocolumn, 
  nofootinbib,showpacs,superscriptaddress,preprintnumbers]{revtex4-2}

\usepackage[utf8]{inputenc}
\usepackage{graphicx, color}

\usepackage[svgnames,psnames]{xcolor}
\usepackage[colorlinks,citecolor=DarkGreen,linkcolor=FireBrick,linktocpage]{hyperref}

\usepackage{dcolumn}
\usepackage{bm}
\usepackage{amsmath}
\usepackage{amsthm}
\usepackage{ascmac}
\usepackage{xparse} 
\usepackage{mathtools} 
\usepackage{blkarray} 
\usepackage{amscd,verbatim}
\usepackage{amssymb}
\usepackage[all]{xy}
\usepackage{tikz}
\usepackage{tikz-cd}
\usepackage{pgf,pgfplots}
\usepackage{enumitem}

\usepackage{caption}
\usepackage{subcaption}
\usepackage{floatrow}

\newcommand\coker{\mathop{{\rm coker}}}

\newcommand\p{\partial}

\usepackage[normalem]{ulem}  % \sout{old text} for strikeout

\renewcommand\sout{\bgroup \color{red} \ULdepth=-.5ex \ULset}

\newcommand{\be}{\begin{equation}}
\newcommand{\ee}{\end{equation}}

\newcommand{\bea}{\begin{eqnarray}}
\newcommand{\eea}{\end{eqnarray}}

\newcommand{\bsp}{\begin{split}}
\newcommand{\espl}{\end{split}}

\newcommand{\bpm}{\begin{pmatrix}}
\newcommand{\epm}{\end{pmatrix}}

\definecolor{mydarkred}{RGB}{233,20,35}
\definecolor{mypurple}{RGB}{120, 35, 160}
\definecolor{mydarkpurple}{RGB}{128, 100, 162}
\definecolor{mybrown}{RGB}{255, 195, 0}
\definecolor{myaqua}{RGB}{29, 153, 168}
\definecolor{myblue}{RGB}{91, 129, 184}  
\definecolor{mygreen}{RGB}{155, 187, 89}  

\definecolor{mybrightblue}{RGB}{0, 140, 255}  

\tikzstyle{species}=[
  circle,draw=black!100, thick, inner sep=0pt,minimum size=6mm
  ] 
\tikzstyle{reaction}=[rectangle,draw=black!100,fill=black!15,thick, inner sep=0pt,minimum size=6mm]

\theoremstyle{remark}

\newtheorem{lemma}{Lemma}%[subsection]

\newtheorem{theorem}{Theorem}%[section]

\newtheorem*{thm*}{Theorem}

\newtheorem{proposition}[lemma]{Proposition}

\newtheorem{corollary}[lemma]{Corollary}

\newtheorem{assumption}[lemma]{Assumption}

\makeatletter 
    
\renewcommand\onecolumngrid{% <<<<<<
\do@columngrid{one}{\@ne}%
\def\set@footnotewidth{\onecolumngrid}% <<<<<<<<<<<<<<<<
\def\footnoterule{\kern-6pt\hrule width 1.5in\kern6pt}%
}

\renewcommand\twocolumngrid{% <<<<<<
        \def\footnoterule{% restore rule
        \dimen@\skip\footins\divide\dimen@\thr@@
        \kern-\dimen@\hrule width.5in\kern\dimen@}
        \do@columngrid{mlt}{\tw@}
}%

\makeatother

\begin{document}

\title{ 
Robust Perfect Adaptation of Reaction Fluxes Ensured by Network Topology
}

\author{Yuji~Hirono}
\email{yuji.hirono@gmail.com}
\affiliation{Asia Pacific Center for Theoretical Physics, Pohang 37673, Republic of Korea}
\affiliation{Department of Physics, POSTECH, Pohang 37673, Republic of Korea}
\affiliation{
RIKEN iTHEMS, RIKEN, Wako 351-0198, Japan
}

\author{Hyukpyo~Hong}
\email{hphong@kaist.ac.kr}
\affiliation{
Department of Mathematical Sciences, Korea Advanced Institute of Science and Technology, Daejeon 34141, Republic of Korea
}
\affiliation{
Biomedical Mathematics Group, Pioneer Research Center for Mathematical and Computational Sciences, Institute for Basic Science, Daejeon 34126, Republic of Korea
}

\author{Jae~Kyoung~Kim}
\email{jaekkim@kaist.ac.kr}
\affiliation{
Department of Mathematical Sciences, Korea Advanced Institute of Science and Technology, Daejeon 34141, Republic of Korea
}
\affiliation{
Biomedical Mathematics Group, Pioneer Research Center for Mathematical and Computational Sciences, Institute for Basic Science, Daejeon 34126, Republic of Korea
}

\date{\today}

\begin{abstract}
Maintaining stability in an uncertain environment is essential for proper functioning of living systems. Robust perfect adaptation (RPA) is a property of a system that generates an output at a fixed level even after fluctuations in input stimulus without fine-tuning parameters, and it is important to understand how this feature is implemented through biochemical networks. 
The existing literature has mainly focused on RPA of the concentration of a chosen chemical species, and no generic analysis has been made on RPA of reaction fluxes, that play an equally important role.
Here, we identify structural conditions on reaction networks under which all the reaction fluxes exhibit RPA against the perturbation of the parameters inside a subnetwork. Based on this understanding, we give a recipe for obtaining a simpler reaction network, from which we can fully recover the steady-state reaction fluxes of the original system. This helps us identify key parameters that determine the fluxes and study the properties of complex reaction networks using a smaller one without losing any information about steady-state reaction fluxes. 
\end{abstract}

\maketitle

\emph{Introduction.}--How to keep a stable status in a changing environment is a vital issue for biological systems~\cite{kitano2004biological,kitano2007towards}. 
One strategy that living cells adopt for maintaining stability is robust perfect adaptation (RPA)~\cite{barkai1997robustness,alon1999robustness,FERRELL201662,aoki2019universal,KHAMMASH2021509}, which is a property of a system that maintains the levels of certain quantities by counteracting the effect of disturbances inside biochemical reaction networks~\cite{10.1093/nar/28.1.27,jeong2000large,ravasz2002hierarchical}. 
From a control-theoretical viewpoint, RPA can be achieved through integral feedback control~\cite{yi2000robust,8619101,KHAMMASH2021509}, 
and there have been studies on how to implement this within biochemical reaction networks
for a chosen set-point concentration~\cite{araujo2018topological,wang2021structure,Gupta2022.02.01.478605}. 
As an alternative perspective, a topological criterion is developed to identify subnetworks whose parameters are irrelevant to the steady-state properties of the rest of the network~\cite{PhysRevLett.117.048101,PhysRevE.96.022322,PhysRevResearch.3.043123} (see also Refs.~\cite{MOCHIZUKI2015189,doi:10.1002/mma.3436,doi:10.1002/mma.4668}).
In other words, the rest of the network exhibits RPA 
with respect to the perturbation of reaction parameters inside subnetworks that satisfy certain topological conditions. 

So far, the study of RPA has focused mainly on the adaptation of the concentration of a fixed chemical species. However, in many biological contexts, reaction fluxes play an equally important role. For example, they are used as a measure of biomass production or cell growth~\cite{orth2010flux}. 
Despite this, a comprehensive analysis of RPA of reaction fluxes in relation to the underlying network structure has so far been elusive. 
In this Letter, we uncover the topological criterion for identifying the parameters to which {\it all} the reaction fluxes exhibit RPA for generic chemical reaction systems.
We further show that the steady-state reaction fluxes can be reconstructed from those of a simpler reduced system, which is easier to analyze. 
Namely, the reduced system is the smallest faithful representation of the original system in the sense that the steady-state reaction fluxes can be fully reconstructed. 
We will demonstrate the method with simple hypothetical examples as well as a realistic network of the metabolic pathways of {\it Escherichia coli} ({\it E.~coli}). 
Our finding allows us to identify the parameters that are relevant in determining the reaction fluxes in a possibly complex reaction network and helps us to understand the behavior of the system. 
As an advantage of the present method, we stress that we do not assume any particular reaction kinetics; hence, the results are broadly applicable to generic reaction systems.

\emph{Chemical reaction systems.}--We consider a deterministic chemical reaction system based on a reaction network $\Gamma = (V, E)$, where $V$ and $E$ are sets of chemical species and reactions, respectively. Reaction $e_A \in E$ can be specified in the form 
\begin{equation}
  e_A :
 \sum_i y_{iA} v_i \to 
 \sum_i \bar y_{iA} v_i , 
\end{equation}
where $v_i \in V$, 
and $y_{iA}$ and $\bar y_{iA}$ are stoichiometric coefficients of species $v_i$. 
We denote the stoichiometric matrix by 
$S$, whose components are given by $S_{iA} \coloneqq \bar y_{iA} - y_{iA}$. Using $S$, we can write down the rate equations that describe the time evolution of chemical concentrations:
\begin{equation}
  \frac{d}{dt} \bm x (t) =  S \bm r , 
  \label{eq:rate}
\end{equation}
where $x_i$ is the concentration of $v_i$ and $r_A$ is the reaction rate (i.e., fluxes) of reaction $e_A$. 
To solve the rate equations, we need to express the reaction fluxes, $r_A$, as functions of reactant concentrations, $\bm x$, and parameters, $k_A$, and this choice is called {\it kinetics}, i.e., $r_A = r_A (\bm x; k_A)$. 

We consider a situation where the system reaches an asymptotically-stable steady state in the long-time limit, and examine how the system responds to changes in parameters. 
The steady-state solution 
can be obtained by solving $S \bm r =\bm 0$. 
When the stoichiometric matrix 
has a nontrivial cokernel (left null space), the system has conserved charges, and we can specify their values to obtain the steady-state solution~\cite{FAMILI200316}: if we pick a basis $\{\bm d^{(\bar\alpha)} \}_{\bar\alpha}$ of $\coker S$, 
they can be fixed by 
\begin{equation}
\ell^{\bar\alpha} 
= 
\bm d^{(\bar\alpha)} \cdot \bm x,
\label{eq:cons}
\end{equation}
where $\ell^{\bar\alpha} \in \mathbb R$ specifies the value of the conserved charge $\bm d^{(\bar\alpha)} \cdot \bm x$. 

\emph{Robust perfect adaptation of reaction fluxes.}--We find that there is a simple topological criterion for identifying the parameters to which all the reaction fluxes exhibit RPA.
The topological condition is expressed by an index, which we will define below. 

Let us choose a subnetwork $\gamma = (V_\gamma, E_\gamma)$ of $\Gamma = (V, E)$, i.e., $V_\gamma \subset V$ and $E_\gamma \subset E$.
If $E_\gamma$ includes all the reactions whose reactants are in $V_\gamma$, 
we call the subnetwork to be {\it output-complete}. 
For an output-complete subnetwork $\gamma$, 
we introduce the {\it flux influence index}, 
which is an integer given by 
\begin{equation}
\lambda_{\rm f} (\gamma) \coloneqq  
- |V_\gamma| + |E_\gamma| 
+ |P^0_\gamma (\coker S)|. 
\label{eq:f-index}
\end{equation}
The first two terms 
are the number of chemical species and reaction in the subnetwork, 
respectively,
and the last term is the dimension 
of projected $\coker S$ to the species in $\gamma$:
\begin{equation}
   P^0_\gamma (\coker S) 
   \coloneqq 
   \left\{
   P^0_\gamma \bm d 
   \, \middle| \,
   \bm d \in \coker S 
   \right\}, 
\end{equation}
where $P^0_\gamma$ is the projection matrix to the species in $\gamma$.
Intuitively,
$|P^0_\gamma (\coker S)|$ counts the number of conserved charges 
that include the species inside $\gamma$. The flux influence index $\lambda_{\rm f}(\gamma)$ 
can be used to identify irrelevant subsystems for determining the steady-state reaction fluxes as follows\footnote{We provide the proofs of Theorems 1 \& 2 in the Supplemental Material (SM).}.
\begin{theorem}[RPA of reaction fluxes]  \label{th:rpa-flux}
Let $\gamma \subset \Gamma$ 
be an output-complete subnetwork of a chemical reaction network $\Gamma$. 
When $\gamma$ satisfies $\lambda_{\rm f}(\gamma)=0$, steady-state reaction fluxes 
of $\Gamma$ do not change under the variation of rate parameters of the reactions inside $\gamma$ 
and the values of conserved charges that have nonzero support in $\gamma$. 
Namely, the steady-state reaction flux 
$\bar r_A$ of {\it any} reaction in $\Gamma$, $e_A \in E$, satisfies 
$\frac{\p}{\p k_{B^\ast}} \bar r_A = 0$ 
where $e_{B^\ast} \in E_\gamma$
and 
$\frac{\p}{\p \ell_{\bar\alpha^\ast}} \bar r_A = 0$ 
where $\ell_{\bar\alpha^\ast}$ is the value of a conserved charge with nonzero support in $\gamma$. 
\end{theorem}
This means that reaction fluxes 
of all the reactions in $\Gamma$ 
exhibit RPA 
with respect to the change of reaction parameters or values of conserved charges 
inside $\gamma$ if $\lambda_{\rm f}(\gamma)=0$. 
Notably, the adaptation is robust, meaning that it does not require any fine-tuning of parameters.
This robustness is a consequence of the fact that the condition of vanishing index is  determined only by the topology of the network and is insensitive to the details of the reactions. 
We refer to a subnetwork $\gamma$ with $\lambda_{\rm f}(\gamma)=0$ as a {\it strong buffering structure}. 
A reaction network can contain multiple strong buffering structures. 
As we show in the SM, strong buffering structures are closed under union and intersection: if $\gamma_1$ and $\gamma_2$ are strong buffering structures, so are $\gamma_1 \cup \gamma_2$ and $\gamma_1 \cap \gamma_2$. 
We say that a strong buffering structure is {\it maximal} if it contains all the strong buffering structures in the network. 

\emph{Example 1.}--As a simple example, let us consider the following reaction network $\Gamma$ with four species and six reactions,
\begin{align}
    e_1 &: \text{(input)} \to v_1, \notag \\
    e_2 &: v_1 \to v_2, \notag \\
    e_3 &: v_2 \to v_3 ,\notag \\
    e_4 &: v_3 \to v_4 ,  \\
    e_5 &: v_4 \to v_1 , \notag \\
    e_6 &: v_3 \to \text{(output)} . \notag
\end{align}
See the left part of Fig.~\ref{fig:ex1-red} for the visual representation of the network. 
The network contains seven strong buffering structures (see the SM for details). 
Among them, $\gamma_7^\ast = (\{v_1,v_2,v_4\},\{e_2,e_3,e_5 \})$ has a zero flux influence index ($\lambda_{\rm f}(\gamma) = -3 + 3 + 0 =0$) and is maximal.
Theorem 1 predicts that
reaction fluxes exhibit RPA with respect to parameters $k_2, k_3, k_5$. 
Indeed, if we for example employ the mass-action kinetics
and solve for the steady-state reaction fluxes, 
they only depend on $k_1, k_4,$ and $k_6$:
\begin{equation}
\bar{ \bm r} = 
k_1 \bm c^{(1)} + \frac{k_4 k_1}{ k_6 } \bm c^{(2)},
\label{eq:ex1-rate-ss}
\end{equation}
where 
$\bm c^{(1)} \coloneqq \begin{bmatrix}
1 & 1 & 0 & 1 & 0 & 1     
\end{bmatrix}^\top$ 
and 
$\bm c^{(2)} \coloneqq \begin{bmatrix}
1 & 1 & 1 & 0 & 1 & 0
\end{bmatrix}^\top$ 
are basis vectors of $\ker S$ (the component are arranged in the order of $\{e_2,e_3,e_5,e_1,e_4,e_6 \}$ for later convenience).

\emph{Minimal form.}--Combining the notion of strong buffering structures 
with the reduction method introduced in Ref.~\cite{PhysRevResearch.3.043123}, 
we can define a simpler yet faithful representation of complex reaction networks, in the sense that the steady-state fluxes of the original system can be fully recovered from those of the simplified system. 

Let us describe the construction of a reduced reaction system. 
For a chosen output-complete subnetwork $\gamma$, 
we separate the chemical concentrations and reaction fluxes into those inside/outside $\gamma$ as 
\begin{equation}
    \bm x = 
    \begin{pmatrix}
      \bm x_1 \\
      \bm x_2
    \end{pmatrix}, 
    \quad 
    \bm r = 
    \begin{pmatrix}
      \bm r_1 \\
      \bm r_2
    \end{pmatrix}, 
\end{equation}
where $1$ and $2$ correspond 
to inside and outside degrees of freedom, respectively. 
Accordingly to this separation of species/reactions, the stoichiometric matrix $S$ is partitioned into four blocks as 
\begin{equation}
  S = 
\begin{pmatrix}
  S_{11} & S_{12} \\
  S_{21} & S_{22} 
\end{pmatrix}. 
\end{equation}
In the reduced system, 
the species and reactions inside $\gamma$ 
are eliminated. 
The rate equation of the reduced reaction system is written as 
\begin{equation}
  \frac{d}{dt} {\bm x}_2 
   = S' \bm r_2 (\bm x_2), 
\end{equation}  
where $S'$ is the so-called generalized Schur complement of $S$ with respect to $S_{11}$,
\begin{eqnarray}
S' \coloneqq S_{22}  - S_{21} S_{11}^+ S_{12}. 
\label{eq:def-sp} 
\end{eqnarray}
Here, $S^+_{11}$ is the Moore-Penrose inverse of $S_{11}$. 
We will denote the reduced network obtained by removing $\gamma$ from $\Gamma$ as $\Gamma' \coloneqq \Gamma / \gamma$. 
The structure of the reduced network is characterized by $S'$,
and the second term of Eq.~\eqref{eq:def-sp} is responsible for the reconnections of reactions associated with the removal of subnetwork $\gamma$.

For a given reaction network $\Gamma$, 
its {\it minimal form} $\Gamma_{\rm min}$ 
is a reaction system obtained by reducing 
the maximal strong buffering structure in $\Gamma$
in the sense described above. 
We can show the following:\footnote{
We make a technical assumption on the nature of conserved charges in the system. See the SM for details. 
}
\begin{theorem}[Reconstruction of steady-state reaction fluxes from the minimal form]  \label{th:recon-min}
Let $\Gamma$ be a reaction system and $\Gamma_{\rm min}$ be its minimal form. 
Then, the steady-state reaction fluxes of $\Gamma$ can be reconstructed from those of $\Gamma_{\rm min}$. 
\end{theorem}
We provide a proof in the SM. 

Let us illustrate the reconstruction procedure. 
Suppose that the steady-state reaction flux in $\Gamma_{\rm min}$ is given by 
\begin{equation}
\bar{\bm r}_2 = \sum_\alpha \mu^{(\alpha)} \bm c_2^{(\alpha)} ,
\label{eq:r2-ss}
\end{equation}
where $\{\bm c^{(\alpha)}_2 \}_\alpha$ is 
a basis of $\ker S'$,
and the coefficients 
$\mu^{(\alpha)} = \mu^{(\alpha)}(\bm k_2)$ are 
functions of the reaction parameters outside $\gamma$. 
When $\gamma$ is a strong buffering structure, 
there is an isomorphism $f$
from $\ker S'$ to $\ker S$,
\begin{equation}
\ker  S'
\ni
\bm c_2
\mapsto 
f(\bm c_2) =
\begin{bmatrix}
- S_{11}^+ S_{12} \bm c_2
\\
\bm c_2 
\end{bmatrix} 
\in \ker S.
\label{eq:map-sp-s}
\end{equation}
Using this map, the steady-state reaction flux of the original system is reconstructed by 
\begin{equation}
\bar{\bm r} 
= \sum_\alpha \mu^{(\alpha)}f(\bm c^{(\alpha)}_2) ,
\end{equation}
where the coefficients $\mu^{(\alpha)}$ are 
the same as Eq.~\eqref{eq:r2-ss}.

In Example 1, $\gamma_7^\ast$ 
is maximal, and 
the minimal form is obtained as $\Gamma_{\rm min} = \Gamma / \gamma_7^\ast$. 
Under the reduction $\Gamma \to \Gamma_{\rm min}$, 
the stoichiometric matrix changes as 
\begin{equation}
\resizebox{9cm}{!}{%
\begin{tikzpicture}%[scale=0.7]
\node at (0, 0) {
$
S =
\begin{blockarray}{cccccccc}
&& \\
\begin{block}{c[ccccccc]}
{\color{mydarkred}v_1} \quad\, 
& -1 & 0 & 1 & 1 & 0 & 0 \\
 {\color{mydarkred}v_2} \quad\, 
& 1 & -1 & 0 & 0 & 0 &0 \\
{\color{mydarkred}v_4} \quad\, 
& 0 & 0 & -1 & 0 & 1 &0 \\
v_3 \quad\, 
& 0 & 1 & 0  & 0 & -1 & -1 \\
\end{block} 
& {\color{mydarkred}e_2} 
& {\color{mydarkred}e_3}
& {\color{mydarkred}e_5}
& e_1 
& e_4 
& e_6
\end{blockarray}  
$
}; 
\draw[mydarkred, dashed,line width=1] 
 (-0.9,-0.21) rectangle (0.8,1);
\node at (-0.1, 1.25) {\color{mydarkred}$S_{11}$} ;

\draw [mybrightblue,ultra thick,-latex]   (3,0) -- (4, 0); 

\node at (5.8, 0) {
$
S' =
\begin{blockarray}{ccccc}
&& \\
\begin{block}{c[cccc]}
v_3 \,\, & 1 & 0 & -1 \\
\end{block} 
& e_1 & e_4 & e_6
\end{blockarray}
          $ %\quad .
      }; 

\end{tikzpicture}
}
\label{eq:ex-v4-e6-s-sp}
\end{equation}
where we have brought the components in 
$\gamma_7^\ast$ to the upper-left part. 
The reduction can be visually represented 
as Fig.~\ref{fig:ex1-red}. 

\begin{figure}[tb]
\centering
\resizebox{9cm}{!}{
\begin{tikzpicture} 
    \node[species] (v1) at (1.25,0) {$v_1$}; 
    \node[species] (v2) at (2.5,0) {$v_2$};
    \node[species] (v3) at (3.75,0) {$v_3$};
    \node[species] (v4) at (2.5,1.5) {$v_4$}; 
    \node (d1) at (0,0) {}; 
    \node (d2) at (5,0) {}; 
 
    \draw [-latex,draw,line width=0.5mm] (d1) edge node[below]{$e_1$} (v1);
    \draw [-latex,draw,line width=0.5mm] (v1) edge node[below]{$e_2$} (v2);
    \draw [-latex,draw,line width=0.5mm] (v2) edge node[below]{$e_3$} (v3);
    \draw [-latex,draw,line width=0.5mm] (v3) edge node[above right]{$e_4$} (v4);
    \draw [-latex,draw,line width=0.5mm] (v4) edge node[above left]{$e_5$} (v1);
    \draw [-latex,draw,line width=0.5mm] (v3) edge node[below]{$e_6$} (d2);
       
   \draw [mybrightblue,ultra thick,-latex]
   (5.5,0) -- (6.5,0); 
   
    \node[species] (red_v3) at (8.25,0) {$v_3$};
    \node (red_d1) at (7,0) {}; 
    \node (red_d2) at (9.5,0) {}; 

    \draw [-latex,line width=0.5mm] (red_d1) edge node[below]{$e_1$} (red_v3);

    \draw [-latex,line width=0.5mm,loop above,out=60,in=120,looseness=10] 
(red_v3) edge node[above]{$e_4$} (red_v3);

    \draw [-latex,line width=0.5mm] (red_v3) edge node[below]{$e_6$} (red_d2);

\end{tikzpicture}}
\caption{
Simple example of reduction to the minimal form (Example 1). The left network is the original one, and the right one is the corresponding minimal form. 
}
\label{fig:ex1-red}
\end{figure}
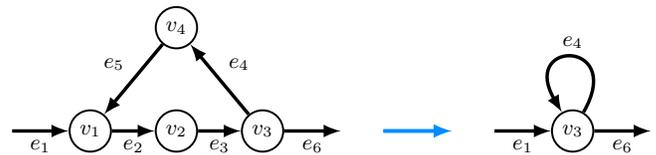

We can easily compute the steady-state reaction fluxes of the original system using its minimal form. 
A set of basis vectors of $\ker S'$ can be picked as
$
\bm c^{(1)}_2 
= 
\begin{bmatrix}
    1 &
    0 &
    1
\end{bmatrix}^\top ,
$
and 
$
\bm c^{(2)}_2 
= 
\begin{bmatrix}
    0 &
    1 &
    0
\end{bmatrix}^\top,
$
where the elements are in the order of $\{e_1,e_4,e_6\}$. 
The rate equation of the reduced system involves only one species and is given by 
\begin{equation}
\frac{d}{dt} x_3 = k_1 - k_6 x_3,
\end{equation}
where we have chosen the mass-action kinetics. 
At the steady state, $\bar x_3 = k_1 / k_6$. 
Since $r_4 = k_4 x_3$, 
its steady-state value is determined as $\bar r_4 = k_4 k_1 / k_6$. 
Hence, the steady-state reaction fluxes in the reduced system are given by
\begin{equation}
\bar{ \bm r}_2 = 
k_1 \bm c^{(1)}_2 + \frac{k_4 k_1}{ k_6 } \bm c^{(2)}_2 . 
\end{equation}
Using the map~\eqref{eq:map-sp-s}, 
the steady-state reaction fluxes of the original system is obtained. 
Indeed, $\bm c_1^{(1)}$ and $\bm c^{(2)}$ in Eq.~\eqref{eq:ex1-rate-ss} are the images of $\bm c_2^{(1)}$
and $\bm c_2^{(2)}$ via the map \eqref{eq:map-sp-s}, 
i.~e. 
$\bm c^{(1)} = f(\bm c^{(1)}_2)$
and 
$\bm c^{(2)} = f(\bm c^{(2)}_2)$

\begin{figure*}
	\centering
		\includegraphics[width= 0.9\textwidth]{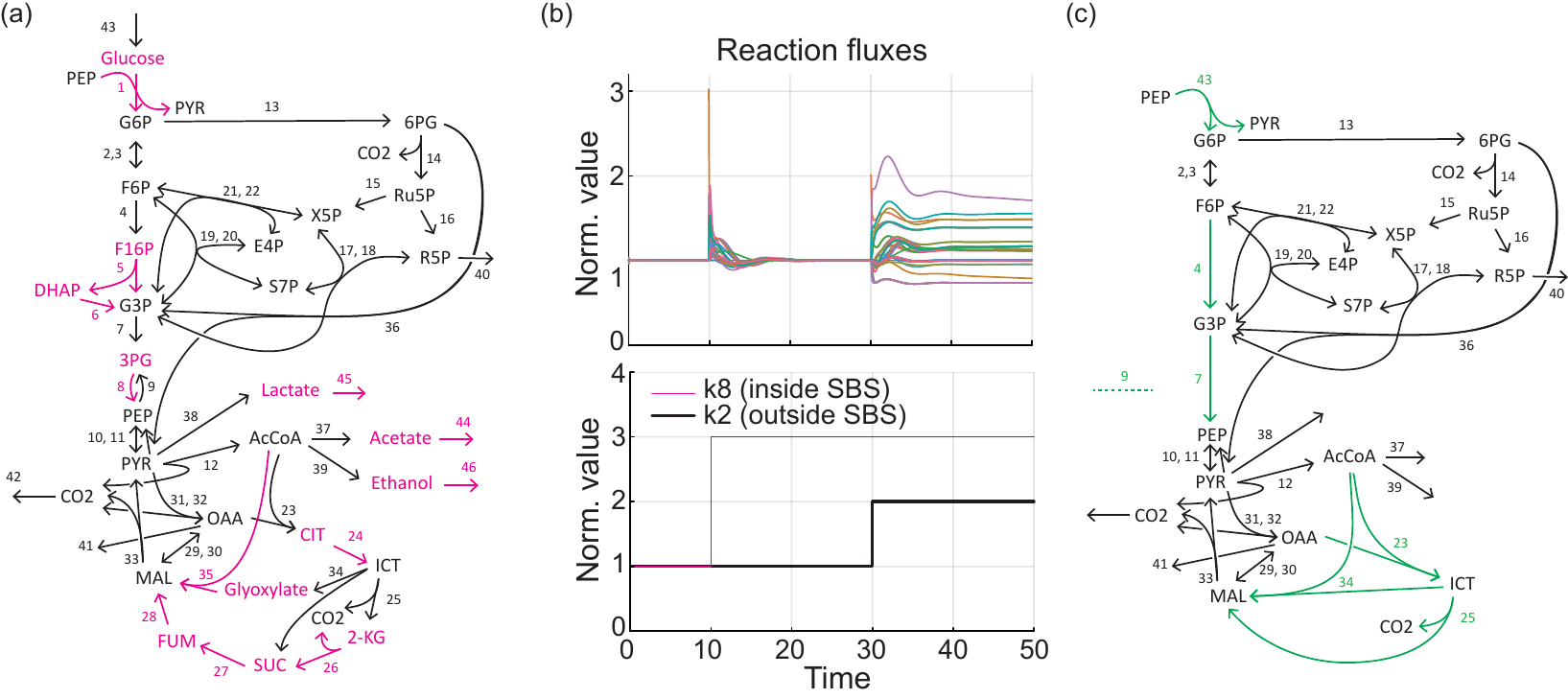}
		\caption{(a) Central metabolic pathways of \textit{E. coli.}. The magenta-colored subnetwork is the maximal buffering structure. 
  (b) Time-series of reaction fluxes of all 46 reactions. After perturbing parameter $k_8$, inside the maximal strong buffering structure, at time $t = 10$, all reaction fluxes exhibit RPA. After perturbing parameter $k_2$, outside the maximal strong buffering structure, at time $t = 30$, the reaction fluxes are changed.
  (c) The minimal form of the network in (a). The green-colored reactions are the reconnected ones after the reduction.}
\label{fig:e.coli} 
\end{figure*}

\emph{Metabolic pathways of E.~coli.}--As a realistic example, here we consider the metabolic pathways \cite{doi:10.1002/mma.3436,ishii2007multiple} of {\it E.~coli}. 
The whole network $\Gamma$ is shown in Fig.~\ref{fig:e.coli}(a), 
which consists of the glycolysis, the pentose phosphate pathway, and the tricarboxylic acid cycle.
This network contains 12 generators of strong buffering structures (note that the unions of them are also strong buffering structures).
The maximal one, $\bar\gamma_{\rm s}$, is the union of all the 12 subnetworks (Fig. 2(a); magenta-colors).
All the reaction fluxes at the steady-state 
are insensitive to 12 parameters in $\bar\gamma_{\rm s}$. For instance, when $k_8$ in $\bar\gamma_{\rm s}$ is perturbed, all the reaction fluxes are adapted to the original steady states (Fig. 2(b)). On the other hand, such RPA does not occur when $k_2$ outside $\bar\gamma_{\rm s}$ is perturbed (Fig. 2(b)).
In Fig.~\ref{fig:e.coli}(c), we present the minimal form $\Gamma_{\rm min}\coloneqq \Gamma / \bar\gamma_{\rm s}$. 
While the original system contains 28 species and 46 reactions, 
the reduced system has 16 species and 34 reactions. 
The steady-state reaction fluxes of $\Gamma$ can be exactly reconstructed from the smaller system $\Gamma_{\rm min}$. The reduced system is simpler to analyze and is helpful for us to identify relevant parameters for determining reaction fluxes, which can be useful for flux balance analysis~\cite{https://doi.org/10.1002/wsbm.60,orth2010flux,10.1093/bib/bbp011,KAUFFMAN2003491}.

\emph{Summary and discussions.}--In this Letter, we have identified the structural condition under which all the reaction fluxes show RPA with respect to the parameters inside a subnetwork $\gamma$: 
this is achieved when $\gamma$ is output-complete and 
has vanishing flux influence index $\lambda_{\rm f}(\gamma) = 0$ introduced in Eq.~\eqref{eq:f-index}. 
Identifying strong buffering structures is particularly useful for the analysis of large reaction networks, as the parameters inside these subnetworks are irrelevant in determining the steady-state reaction fluxes and can be disregarded from the start. 
Theorem~1 provides us with a constraint from the network topology on the response of reaction systems against parameter perturbations and would be helpful in understanding the results of, for example, metabolic control analysis~\cite{kacser1973control,fell1992metabolic,szathmary1993deleterious,maclean2010predicting,bagheri2004evolution}. 
We gave a method for obtaining the corresponding reduced system (minimal form), which is the tightest expression of the original reaction network without losing any information about the steady-state reaction fluxes. 
Since the condition of vanishing index is a topological one and is insensitive to the details of reactions such as reaction kinetics and parameter values, 
the results hold for any kinetics/parameters
and thus are broadly applicable. 

Finally, let us comment on the relation of strong buffering structures to buffering structures defined by the vanishing of the influence index, $\lambda(\gamma)$~\cite{PhysRevLett.117.048101,PhysRevE.96.022322,PhysRevResearch.3.043123}. 
When a subnetwork has a vanishing influence index $\lambda(\gamma)=0$, the concentrations and reaction fluxes outside $\gamma$ exhibit RPA with respect to the parameters inside $\gamma$.
Strong buffering structures are a restricted type of buffering structures, and {\it all} the reaction fluxes show RPA for the parameters inside $\gamma$ if $\lambda_{\rm f}(\gamma)=0$. 
As we show in the SM, a strong one is always a usual buffering structure: $\lambda_{\rm f}(\gamma)=0$ implies $\lambda(\gamma)=0$.

\begin{acknowledgments}
YH is supported by an appointment of the JRG Program at the APCTP, which is funded through the Science and Technology Promotion Fund and Lottery Fund of the Korean Government, and is also supported by the Korean Local Governments of Gyeongsangbuk-do Province and Pohang City. YH is also supported by the National Research Foundation (NRF) of Korea (Grant No. 2020R1F1A1076267) funded by the Korean Government (MSIT).
HH is supported by the National Research Foundation of Korea (NRF) NRF-2019-
Fostering Core Leaders of the Future Basic Science Program/Global Ph.D. Fellowship Program 2019H1A2A1075303.
JKK is supported by the Institute for Basic Science IBS-R029-C3.
\end{acknowledgments}

\onecolumngrid
\newpage

{\centerline{\large\textbf{Supplemental Material for}}}
\vspace{1mm}
{\centerline{\large\textbf{``Robust Perfect Adaptation of Reaction Fluxes Ensured by Network Topology''}}}

\section{ Notations }

In the following, the sets of chemical species and reactions 
are denoted by $V$ and $E$, respectively. 
For a given species $v_i \in V$
and reaction $e_A \in E$, 
the corresponding concentration 
and reaction rate 
will be denoted by $x_i$ and $r_A$, respectively. 
The cardinality of a set $V$ is indicated by $|V|$. 
The dimension of a vector space $W$ is denoted by $|W|$. 
For a given set of vectors, $\bm v_1, \bm v_2, \cdots$, 
 ${\rm span\,}\{ \bm v_1, \bm v_2, \cdots \}$ 
indicates the vector space spanned by them.

\section{ Relation of the flux influence index and influence index }

Here, we discuss the relationship between 
the flux influence index $\lambda_{\rm f}(\gamma)$ 
introduced in this work and the influence index $\lambda(\gamma)$~\cite{PhysRevLett.117.048101SM,PhysRevE.96.022322SM,PhysRevResearch.3.043123SM}. 
These two indices are related by 
\begin{equation}
\lambda_{\rm f} (\gamma) 
= 
\lambda (\gamma) 
+ |(\ker S)_{{\rm supp}\, \gamma}|,
\label{eq:rel-f-index-i-index}
\end{equation}
where the last term is the dimension of the following vector space,
\begin{equation}
 (\ker S)_{{\rm supp}\,\gamma}
 \coloneqq 
 \left\{
  \bm c 
  \, \middle| \,
  \bm c \in \ker S, 
  P^1_\gamma \bm c = \bm c 
  \right\}. 
  \label{eq:ker-s-supp-gamma}
\end{equation}  
Here, $P^1_\gamma$ is the projection matrix to $\gamma$ in the space of chemical reactions. 
Intuitively speaking, $(\ker S)_{{\rm supp}\,\gamma}$ is the space of cycles that are supported inside subnetwork $\gamma$. 
Noting that $\lambda (\gamma)$ is nonnegative
for reaction systems with asymptotically stable steady states, $\lambda_{\rm f}(\gamma)=0$ implies $\lambda(\gamma)=0$. 
Thus, a strong buffering structure is always a buffering structure.

\section{ Closure property of strong buffering structures }

We first prove the submodularity of the flux influence index. It consequently implies that strong buffering structures are closed under the union and intersection. 

\begin{proposition}[Submodularity of the flux influence index]\label{prop:submod}
The following inequality holds for any two subnetworks $\gamma_1$ and $\gamma_2$,
\begin{equation}
\lambda_{\rm f} (\gamma_1) + \lambda_{\rm f} (\gamma_2) 
\geq \lambda_{\rm f} (\gamma_1 \cap \gamma_2)  + \lambda_{\rm f} (\gamma_1 \cup \gamma_2). 
\end{equation}
\end{proposition}
\begin{proof}
By definition, the flux influence indices are given by 
\begin{align}
\lambda_{\rm f} (\gamma_1 \cap \gamma_2)  & =
- |V_{\gamma_1 \cap \gamma_2}| + |E_{\gamma_1 \cap \gamma_2}| 
+ |P^0_{\gamma_1 \cap \gamma_2} (\coker S)|, \\ 
\lambda_{\rm f} (\gamma_1 \cup \gamma_2)  & =
- |V_{\gamma_1 \cup \gamma_2}| + |E_{\gamma_1 \cup \gamma_2}| 
+ |P^0_{\gamma_1 \cup \gamma_2} (\coker S)|.
\end{align}
By summing these two, we get 
\begin{equation}
\begin{split}
    \lambda_{\rm f} (\gamma_1 \cap \gamma_2)  + \lambda_{\rm f} (\gamma_1 \cup \gamma_2)  &=
- (|V_{\gamma_1 \cap \gamma_2}| +|V_{\gamma_1 \cup \gamma_2}|) + |E_{\gamma_1 \cap \gamma_2}| + |E_{\gamma_1 \cup \gamma_2}| 
+ |P^0_{\gamma_1 \cap \gamma_2} (\coker S)|+ |P^0_{\gamma_1 \cup \gamma_2} (\coker S)|  \\ 
& =- (|V_{\gamma_1}| +|V_{\gamma_2}|) + |E_{\gamma_1}| + |E_{\gamma_2}| 
+ |P^0_{\gamma_1 \cap \gamma_2} (\coker S)|+ |P^0_{\gamma_1 \cup \gamma_2} (\coker S)| \\ 
& = \lambda_{\rm f} (\gamma_1) + \lambda_{\rm f} (\gamma_2) - |P^0_{\gamma_1} (\coker S)| - |P^0_{\gamma_2} (\coker S)| + |P^0_{\gamma_1 \cap \gamma_2} (\coker S)| + |P^0_{\gamma_1 \cup \gamma_2} (\coker S)| \\ 
&\leq \lambda_{\rm f} (\gamma_1) + \lambda_{\rm f} (\gamma_2).
\end{split}
\end{equation}
The last inequality holds because $|P^0_{\gamma} (\coker S)|$ is a submodular function, i.e., $|P^0_{\gamma_1} (\coker S)| + |P^0_{\gamma_2} (\coker S)| \geq |P^0_{\gamma_1 \cap \gamma_2} (\coker S)| + |P^0_{\gamma_1 \cup \gamma_2} (\coker S)|$, as shown in a previous work~\cite{PhysRevResearch.3.043123SM}. Consequently, the flux influence index is also a submodular function. 
\end{proof}

\begin{corollary}[Closure property of strong buffering structures]
If $\gamma$ and $\gamma_2$ are strong buffering structures, then both $\gamma_1 \cap \gamma_2$ and $\gamma_1 \cup \gamma_2$ are strong buffering structures. 
\end{corollary}
\begin{proof}
To prove that the intersection and union are strong buffering structures, we need to show that they are output-complete and their flux influence indices are zero. The output-completeness is clearly guaranteed after taking the intersection and union. Since the flux influence indices are always nonnegative, $\lambda_{\rm f} (\gamma_1) = \lambda_{\rm f} (\gamma_2) = 0$ implies that $\lambda_{\rm f} (\gamma_1 \cap \gamma_2) = \lambda_{\rm f} (\gamma_1 \cup \gamma_2) = 0 $ by the submodularity of the flux influence index (Proposition~\ref{prop:submod}), which means that the intersection and union are strong buffering structures.
\end{proof}

\section{ Proof of Theorem 1 }

We assume that the system reaches a steady state in the long time limit, and the chemical concentrations at the steady state are determined uniquely for a given set of rate parameters and the values of conserved charges, $\{k_A, \ell^{\bar\alpha}\}$.
At the steady state, 
the reaction rates and the concentrations satisfy\footnote{
Here we use different characters for 
indices of 
chemical species ($i, j, \ldots$), 
reactions ($A,B,\ldots$), 
the basis of the kernel of $S$ ($\alpha, \beta, \ldots$), and that of cokernel of $S$ ($\bar\alpha, \bar\beta, \ldots$).
We use the notation where $|i|$ indicates the number of values the index $i$ takes. 
Thus, $|i|$ indicates the number of chemical species 
and $|A|$ indicates the number of reactions. 
}
\begin{align}
\sum_A  S_{iA}   r_A ( \bm x (\bm k, \bm \ell), k_A ) &= 0, \label{eq:sr} \\
 \sum_i d^{(\bar \alpha)}_i x_i (\bm k, \bm \ell)
 &= \ell^{\bar\alpha},  \label{eq:l-dx} 
\end{align} 
where 
$\{\bm d^{(\bar\alpha)}\}_{\bar\alpha=1,\ldots ,|\bar\alpha|}$ is a basis of $\coker S$,
and 
the second equation specifies the values of conserved charges. 
We assume the existence of an asymptotically stable steady state, and we study how the state 
changes under the change of parameters. 
The concentrations and reaction rates 
are determined by solving Eqs.~\eqref{eq:sr} and \eqref{eq:l-dx}. 
Note that steady-state reaction rates $r_A (\bm x (\bm k, \bm \ell), k_A)$ have an explicit dependence on $k_A$, as well as an implicit dependence on $\bm k$ and $\bm \ell$ through steady-state concentrations, $x_i (\bm k, \bm \ell)$. 
Equation~\eqref{eq:sr} indicates that the steady-state reaction rates are in the kernel of 
matrix $S$. Hence, the rates can be written as 
\begin{equation}
  r_A ( \bm x (\bm k, \bm \ell), k_A ) 
  = \sum_\alpha \mu_\alpha (\bm k, \bm \ell) c^{(\alpha)}_A . 
  \label{eq:r-mu-c}
\end{equation}
where $\{\bm c^{(\alpha)}\}_{\alpha=1,\ldots ,|\alpha|}$ 
is a basis of $\ker S$. 
Taking the derivative of 
Eqs.~(\ref{eq:r-mu-c}) and (\ref{eq:l-dx}) 
with respect to $k_B$ and $\ell^{\bar\beta}$, 
we have
\begin{align}
 \sum_i \frac{\p r_A}{\p x_i} 
  \frac{\p x_i}{\p k_B} 
  + 
  \frac{\p r_A}{\p k_B}
  &=  \sum_\alpha \frac{\p \mu_\alpha }{\p k_B} c^{(\alpha)}_A  , 
  \label{eq:der-1}
  \\
  \sum_i \frac{ \p r_A }  {\p x_i } 
 \frac{\p x_i}{\p \ell^{\bar\beta }} 
 &=  
 \sum_\alpha  \frac{ \p \mu_\alpha }  {\p \ell^{\bar\beta }} c^{(\alpha)}_A , \\
\sum_i d_i^{(\bar\alpha)} \frac{\p x_i}{\p k_B} &= 0,  
 \\
\sum_i  d_i^{(\bar\alpha)} \frac{\p x_i}{\p \ell^{\bar\beta}}
 &=  \delta^{\bar\alpha \bar\beta} . 
\label{eq:der-4}
\end{align}
Here, $\frac{\p r_A}{\p k_B}$ indicates the derivative with respect to the {\it explicit} dependence, 
and 
$\frac{\p r_A}{\p x_i}$ is evaluated at the steady state. 
We introduce a square matrix by 
\begin{equation}
  {\bf A} 
  \coloneqq 
  \begin{pmatrix}
    \p_i r_A 
    & - c^{(\alpha)}_A \\
    d^{(\bar\alpha)}_i & \bm 0 
  \end{pmatrix},
\end{equation}
where 
$\p_i \coloneqq \p / \p x_i$. 
Using this, Eqs.~(\ref{eq:der-1}--\ref{eq:der-4}) are summarized in a compact way in matrix form, 
\begin{equation}
  {\bf A}  
  \begin{pmatrix}
 \p_B x_i \\
  \p_B \mu_\alpha     
  \end{pmatrix}
  = - 
  \begin{pmatrix}
  \p_B  r_A \\
    \bm 0
\end{pmatrix}
, 
\quad 
{\bf A} \, 
\begin{pmatrix}
\p_{\bar\beta} x_i \\
\p_{\bar\beta} \mu_\alpha     
\end{pmatrix}
= 
\begin{pmatrix}
 \bm 0 \\
 \delta^{\bar\alpha \bar\beta }
\end{pmatrix}, 
\label{eq:sensitivity-1}
\end{equation}
where 
$\p_B \coloneqq \partial /\partial k^B$ 
,
$\p_{\bar\beta} \coloneqq \partial /\partial \ell^{\bar\beta}$.
By multiplying ${\bf A}^{-1}$ on Eq.~\eqref{eq:sensitivity-1}, we have
\begin{equation}
 \p_B
 \begin{pmatrix}
  x_i \\
  \mu_\alpha     
  \end{pmatrix}
  = - 
  {\bf A}^{-1}
  \begin{pmatrix}
  \p_B  r_A \\
    \bm 0
\end{pmatrix}
, 
\quad 
\p_{\bar\beta}
\begin{pmatrix}
 x_i \\
 \mu_\alpha     
\end{pmatrix}
=  {\bf A}^{-1}
\begin{pmatrix}
 \bm 0 \\
 \delta^{\bar\alpha \bar\beta }
\end{pmatrix}.  
\label{eq:sensitivity-2}
\end{equation}
Note that $\p_B r_A$ is a diagonal matrix, i.e., $\p_B r_A \propto \delta_{BA}$. 
If we partition ${\bf A}^{-1}$ as 
\begin{equation}
  {\bf A}^{-1} = 
  \begin{pmatrix}
    ({\bf A}^{-1})_{iA} & ({\bf A}^{-1})_{i \bar\alpha  } \\
    ({\bf A}^{-1})_{ \alpha A } & ({\bf A}^{-1})_{ \alpha \bar\alpha }  
  \end{pmatrix}, 
\end{equation}
the responses of steady-state concentrations and reaction rates to the perturbations of $k_B$ and $\ell^{\bar\beta}$ are 
proportional to the following components, 
\begin{equation}
  \p_B x_i \propto  ({\bf A}^{-1})_{iB} , 
  \quad 
  \p_{\bar\beta} x_i \propto  ({\bf A}^{-1})_{i \bar\beta } , 
  \quad 
  \p_{B} \mu_{\alpha} \propto ({\bf A}^{-1})_{\alpha B} ,
  \quad 
  \p_{\bar\beta} \mu_{\alpha} \propto ({\bf A}^{-1})_{\alpha \bar\beta }.
  \label{eq:response}
\end{equation}

\begin{figure}[tb]
  \centering
  \includegraphics
%  [keepaspectratio, scale=0.45]
  [clip, trim=0cm 3cm 0cm 4cm, scale=0.45]
  {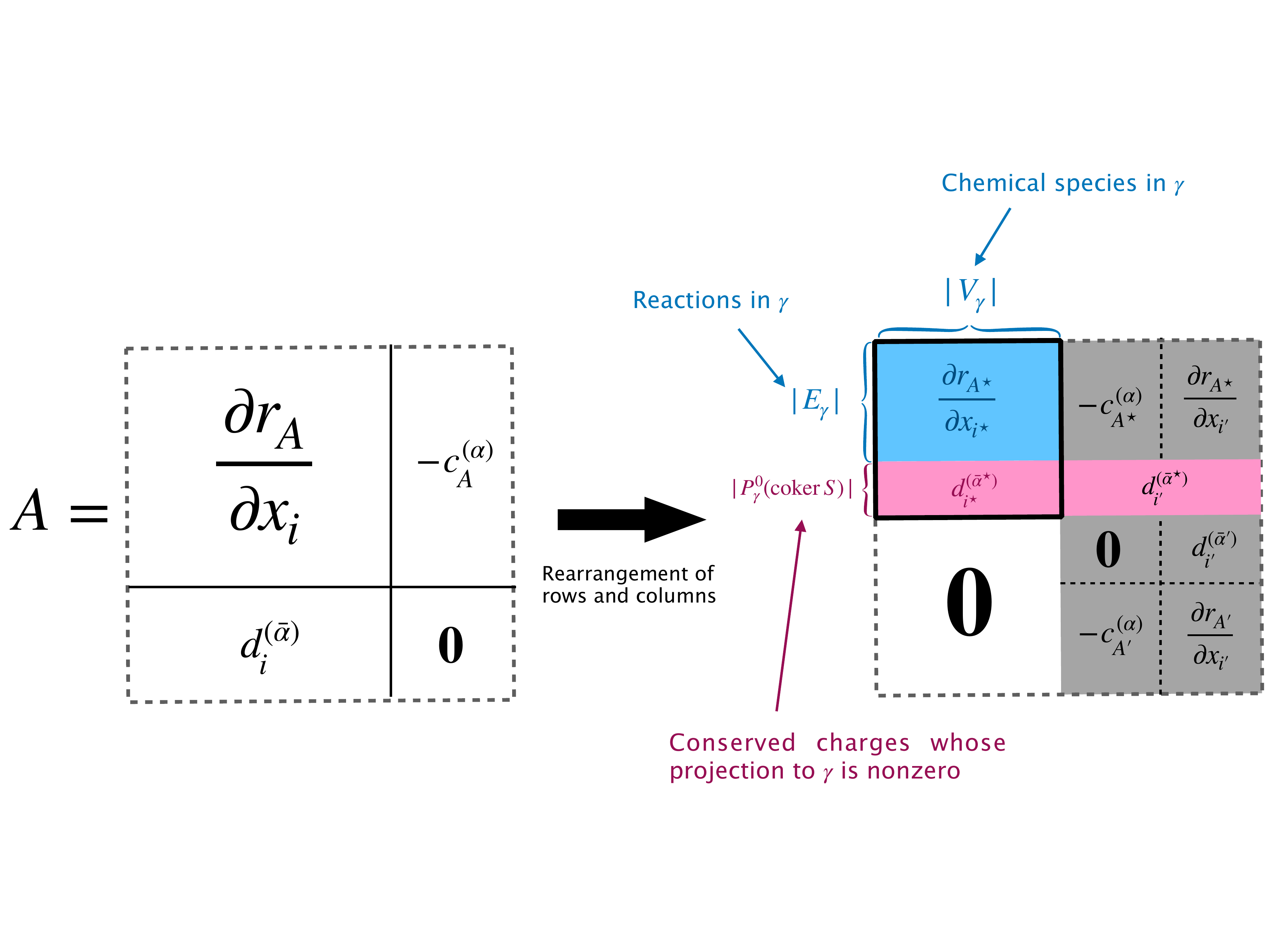} 
  \caption{
Structure of the A-matrix before and after a rearrangement of rows and columns. Indices with stars are those inside $\gamma$, 
and indices with primes are associated with $\Gamma \setminus \gamma$. 
The vectors $\bm d^{({\bar\alpha}^\star)}$ 
are the elements of $\coker S$ whose projections to subnetwork $\gamma$ are nonzero. 
  }
  \label{fig:mat-a}
\end{figure}

For a given output-complete subnetwork $\gamma$, we can bring the rows and columns associated with $\gamma$ to the upper-left corner (see Fig.~\ref{fig:mat-a}). 
All the component of the lower-left part is zero, 
because the reaction rate $r_A$ outside $\gamma$ does not depend on the chemical species in $\gamma$ 
(since $\gamma$ is output-complete). 
When $\lambda_{\rm f}(\gamma) =0$, 
the black box in the upper-left corner is a square matrix. 
Then, ${\bf A}^{-1}$ inherits the same structure, 
\begin{equation}
  {\bf A}^{-1}
  = 
  \begin{pmatrix}
    * & * \\
    \bm 0 & * 
  \end{pmatrix}.
\end{equation}
Namely, if we denote the generic index of ${\bf A}^{-1}$ as $\mu,\nu, \cdots$ and 
write the index inside and outside $\gamma$ 
as $\mu^\star$ and $\mu'$, respectively, 
we have $({\bf A}^{-1})_{\mu' \nu^\star} = 0$. 
Because of this structure, 
\begin{equation}
 \p_{A^\star} \mu_{\alpha} 
 \propto  ({\bf A}^{-1})_{\alpha A^\star} = 0 , 
\quad 
 \p_{{\bar\alpha}^\star} \mu_{\alpha} 
 \propto  ({\bf A}^{-1})_{\alpha {\bar\alpha}^\star} = 0 , 
\end{equation}
for {\it any} $\alpha$. 
Thus, all the steady-state reaction rates are independent of 
the parameters inside strong buffering structures
and the values of conserved charges that have nonzero support inside $\gamma$.

\section{ Proof of Theorem 2 }

We make the following technical assumptions on the nature of conserved charges: 
\begin{assumption} 
For a given element $\bm d$ of $\coker S$, 
let us separate $\bm d$ into 
those inside/outside 
the maximal strong buffering structure $\gamma$ as 
$\bm d = \begin{bmatrix}
\bm d_1 \\
\bm d_2
\end{bmatrix}
$. 
Then, if $\bm d_1 \in \coker S_{11}$,
$
\begin{bmatrix}
\bm d_1 \\
 \bm 0
\end{bmatrix}
\in \coker S
$
holds. 
\label{assumption:charge}
\end{assumption}
Namely, a local conserved charge in a subnetwork $\gamma$ (i.e.\ an element of $\coker S_{11}$) is always a conserved charge in the whole network $\Gamma$. 
We here give a proof that Theorem~2 holds under Assumption~\ref{assumption:charge}. 
In Sec.~\ref{sec-sm:les}, we introduce a long exact sequence which will be used in the proof. 
We will complete the proof in Sec.~\ref{sec:recon-proof}.

\subsection{ Long exact sequence of a pair of chemical reaction networks }\label{sec-sm:les}

To prove Theorem~2, we construct a reduced system 
using the method developed in Ref.~\cite{PhysRevResearch.3.043123SM} 
and show that the steady-state fluxes of the original system can be obtained from the reduced one. 
To capture the changes of cycles and conserved charges associated with a reduction, we use a long exact sequence introduced in Ref.~\cite{PhysRevResearch.3.043123SM}.

Let us introduce its definition. 
We first define the spaces of chains by  
\begin{align}
  C_0 (\Gamma) &\coloneqq 
  \big\{
  \sum_i d_i v_i
  \, |  \, 
  v_i \in V, \, d_i \in \mathbb R 
  \big\} , 
  \\
  C_1 (\Gamma) &\coloneqq 
  \big\{
\sum_A c_A e_A 
  \, |  \,
  e_A \in E, \, c_A \in \mathbb R 
  \big\} .  
\end{align} 
Similarly, 
we define the spaces of chains for a subnetwork $\gamma = (V_\gamma, E_\gamma) \subset \Gamma$, $C_n (\gamma)$ for $n=0,1$, 
to be those generated by $V_\gamma$ and $E_\gamma$. 
We also define the chain spaces of $\Gamma' = \Gamma / \gamma$,
to be the spaces spanned by $(V \setminus V_\gamma, E \setminus E_\gamma)$. 
We will denote the elements of 
$C_1(\gamma)$, 
$C_1(\Gamma)$ 
and 
$C_1(\Gamma')$ 
using a vector,
where each component represents the corresponding coefficient. 
For example, 
\begin{equation}
\bm c_1 \in C_1(\gamma), 
\quad 
\begin{bmatrix}
    \bm c_1 \\
    \bm c_2
\end{bmatrix}
\in C_1(\Gamma), 
\quad 
\bm c_2 \in C_1(\Gamma'), 
\end{equation}
In the above, the element of $C_1(\Gamma)$ are partitioned 
into those inside/outside $\gamma$. 
A similar notation will be used for $C_0(\gamma)$ and so on. 
We consider the following short exact sequence of chain complexes, 
\begin{equation}
  \xymatrix{
    &
    0 
    \ar[d]
    &
    0 \ar[d]
    & 
    0 \ar[d]
    \\
   0 \ar[r] &
   C_1 (\gamma)
%   \ker \varphi_1
   \ar[r]^{\psi_1 }
   \ar[d]^{\p_\gamma } 
   &
   C_1 (\Gamma) 
   \ar[r]^{\varphi_1}
   \ar[d]^{\p }
   & 
   C_1(\Gamma' )
   \ar[r]
   \ar[d]^{\p' }
   & 0
   \\
0
 \ar[r] 
&  
C_0(\gamma)
%\ker \varphi_0 
   \ar[r]^{\psi_0 }
   \ar[d]   
   & 
C_0 (\Gamma) 
  \ar[r]^{\varphi_0}
  \ar[d] 
   &
C_0 (\Gamma' )
   \ar[r]
   \ar[d]
   &
   0
   \\
   & 
   0 
   & 0 
   & 0 
 }
 \label{eq:diag-b}
\end{equation}
where the columns are the chain complexes of $\gamma$, 
$\Gamma$, and $\Gamma'$, respectively. 
The boundary maps are given by the stoichiometric matrices of subnetwork $\gamma$, 
the whole network, and the reduced network\footnote{Recall that the stoichiometric matrix of the reduced system is given by the generalized Schur complement, $S'\coloneqq S_{22} - S_{21}S_{11}^+ S_{12}$, as discussed in the main text}, 
respectively, as
\begin{equation}
  \p_\gamma: \bm c_1 \mapsto S_{11} \bm c_1, 
  \quad 
  \p: \bm c = 
  \begin{bmatrix}
    \bm c_1 \\
    \bm c_2
  \end{bmatrix}
  \mapsto 
  S \bm c, 
  \quad 
  \p': \bm c_2 \mapsto S' \bm c_2. 
\end{equation}
The horizontal maps are given by 
\begin{equation}
  \psi_1: 
  \bm c_1 \mapsto 
  \begin{bmatrix}
    \bm c_1 \\ 
    \bm 0
  \end{bmatrix}, 
  \quad 
  \varphi_1: 
  \begin{bmatrix}
    \bm c_1 \\ 
    \bm c_2
  \end{bmatrix}
  \mapsto \bm c_2 , 
\end{equation}
\begin{equation}
  \psi_0: 
  \bm d_1 \mapsto 
  \begin{bmatrix}
    \bm d_1 \\ 
    S_{21} S^+_{11} \bm d_1
  \end{bmatrix}, 
  \quad 
  \varphi_0: 
  \begin{bmatrix}
    \bm d_1 \\ 
    \bm d_2
  \end{bmatrix}
  \mapsto \bm d_2 - S_{21}S^+_{11} \bm d_1 . 
\end{equation}
One can check that the diagram \eqref{eq:diag-b} commutes when the following condition is satisfied:\footnote{
Note that we use a notation where a unit matrix is denoted by $1$. 
} 
\begin{equation}
  S_{21} (1 - S_{11}^+ S_{11}) \bm c_1 = \bm 0, 
  \label{eq:com-1}
\end{equation}
where $\bm c_1 \in C_1(\gamma)$. 
The matrix $ (1 - S_{11}^+ S_{11})$ 
is the projection matrix to $\ker S_{11}$, 
and Eq.~(\ref{eq:com-1}) is equivalent to 
\begin{equation}
  \ker S_{11} \subset \ker S_{21}  . 
  \label{eq:com-2}
\end{equation}
As shown around Eq.~(182) of Ref.~\cite{PhysRevResearch.3.043123SM}, 
this condition is equivalent to 
$\widetilde c (\gamma) = |\ker S_{11} / (\ker S)_{{\rm supp}\, \gamma}| = 0$.  Note that $\ker S_{11}$ and $(\ker S)_{{\rm supp}\, \gamma}$ are the set of cycles in $\gamma$ and cycles in $\Gamma$ with components belongs to $\gamma$ (see Eq.~\eqref{eq:ker-s-supp-gamma} for details), respectively. Thus,  the number $\widetilde c(\gamma)$ counts emergent cycles, which are cycles in $\gamma$ that cannot be extended to cycles in $\Gamma$
(see Ref.~\cite{PhysRevResearch.3.043123SM} for more details.)
Thus, the diagram \eqref{eq:diag-b} commutes if and only if $\gamma$ has no emergent cycles. 

Assuming Eq.~\eqref{eq:com-2} holds, 
we can apply the snake lemma to the diagram~\eqref{eq:diag-b} and obtain a long exact sequence\footnote{
Note that boundary operators are given by multiplications of stoichiometric matrices, 
so we have 
$\ker \p_\gamma = \ker S_{11}$, 
$\ker \p = \ker S$, and so on. 
}, \\
\begin{equation}
  \xymatrix{
   0 \ar[r] &
   \ker S_{11}
   \ar[r]^{\psi_1} 
   &
   \ker S
   \ar[r]^{\varphi_1}
   &
   \ker S'
   \ar[r]^{\delta_1\quad } &  
   \coker S_{11}
   \ar[r]^{\bar{\psi}_0}
   & 
   \coker S
   \ar[r]^{\bar{\varphi}_0} 
   &
   \coker S'
   \ar[r]
   &
   0 , 
 }  
 \label{eq:les}
\end{equation}
\\ 
where $\bar \psi_0$ and $\bar \varphi_0$ 
are induced maps\footnote{
For example, $\bar \psi_0$ is defined as, 
$\coker S_{11} \ni [\bm d_1]
\mapsto \bar \psi_0 ([\bm d_1]) \coloneqq [\psi_0 (\bm d_1)] \in \coker S$, 
where $[...]$ indicates 
identifying the differences in 
${\rm im\,} S_{11} ({\rm im}\, S)$
for $\coker S_{11} (\coker S)$, respectively.
This map is well-defined in the sense that it does not depend on the choice of a representative $\bm d_1 \in C_0(\gamma)$, 
since if we had picked $\bm d_1 + S_{11} \bm c_1$, 
its image by $\psi_0$ is 
$\psi_0 (\bm d_1 + S_{11} \bm c_1) = 
\psi_0 (\bm d_1) + S 
\begin{bmatrix}
    \bm c_1 \\ \bm 0
\end{bmatrix}
$
and the second term is zero in $\coker S$. 
} of $\psi_0$ and $\varphi_0$. 
The map $\delta_1: \ker S' \to \coker S_{11}$ 
is called the connecting map. 
For a given $\bm c_2 \in \ker S'$, 
the connecting map is given by\footnote{
The  map is identified as follows. 
Pick an element $\bm c_2 \in \ker S'$, which 
can be included in $C_1(\Gamma')$. 
Since $\varphi_1$ is surjective, there exists 
$
\bm c = 
\begin{bmatrix}
  \bm c_1 \\
  \bm c_2 
\end{bmatrix}$ 
such that $\varphi_1 (\bm c) = \bm c_2$. 
From the commutativity of the diagram (\ref{eq:diag-b}), 
we have $\varphi_0 (S\bm c) = S' \bm c_2 = \bm 0$. 
Since the rows of the diagram~(\ref{eq:diag-b}) are exact, 
there exists $\bm d_1 \in C_0(\gamma)$ 
such that $\psi_0 (\bm d_1) = S \bm c \in \ker \varphi_0$. 
We obtain $[\bm d_1]=[S_{11} \bm c_1 + S_{12} \bm c_2]= [S_{12} \bm c_2]\in \coker S_{11}$ by identifying the differences in ${\rm im\,} S_{11}$. 
The mapping $\bm c_2 \mapsto [S_{12} \bm c_2]$ is the connecting map. 
} 
\begin{equation}
\delta_1: 
\bm c_2 \mapsto 
[S_{12} \bm c_2] \in \coker S_{11}, 
\end{equation}
where $[...]$ means to identify the differences in ${\rm im\,} S_{11}$.

\subsection{ Proof of Theorem~2 }\label{sec:recon-proof}

Let us now describe a proof of Theorem~2. 

Firstly, let us rewrite the flux influence index in a form useful for the following proof. 
In Ref.~\cite{PhysRevResearch.3.043123SM}, 
it was shown that the influence index is decomposed as 
\begin{equation}
\lambda(\gamma) = 
\widetilde c(\gamma) + d_l (\gamma) 
- \widetilde d(\gamma),     
\end{equation}
where 
$d_l (\gamma)\coloneqq|(\coker S) / X(\gamma)|
$ with 
\begin{equation}
X (\gamma) 
\coloneqq 
  \left\{ 
    \begin{pmatrix}
      \bm d_1 \\
      \bm d_2
    \end{pmatrix}
    \in \coker S 
    \, \middle| \,
    \bm d_1 \in \coker S_{11}
    \right\},
\end{equation}
and $\widetilde d(\gamma) \coloneqq
|(\coker S_{11}) / D_{11}(\gamma)|$ with
\begin{equation}
  D_{11}(\gamma)
  \coloneqq 
  \left\{ 
    \bm d_1 \in \coker S_{11} 
  \, \middle| \,
  {}^{\exists} \bm d_2 {\text{ such that }} 
  \begin{bmatrix}
    \bm d_1 \\
    \bm d_2
  \end{bmatrix}
  \in \coker S 
  \right\} . 
\end{equation}
Using the relation \eqref{eq:rel-f-index-i-index} and 
noting that $\widetilde c (\gamma) = |\ker S_{11} / (\ker S)_{{\rm supp}\, \gamma}|$, 
we find that the flux influence index is written in the form 
\begin{equation}
  \lambda_{\rm f} (\gamma)
  = |\ker S_{11}| + d_l (\gamma)  - \widetilde d (\gamma) . 
  \label{eq:lambda-f-decom}
\end{equation}
The number $\widetilde d(\gamma)$ 
counts emergent conserved charges in $\gamma$, 
which are conserved only within $\gamma$ and not in $\Gamma$. 
Due to Assumption~\ref{assumption:charge}, 
we have $\widetilde d(\gamma) =0$ and there is no emergent conserved charge
in the current setting. 
Since $|\ker S_{11}|$ and $d_l(\gamma)$ are nonnegative, 
we have $|\ker S_{11}| = d_l(\gamma)=0$. 
Then, $\widetilde c (\gamma) = |\ker S_{11} / (\ker S)_{{\rm supp}\, \gamma}| = 0$
and the diagram \eqref{eq:diag-b} commutes, and the long exact sequence \eqref{eq:les} holds. 
Since $\widetilde d (\gamma)=0$, 
the connecting map $\delta_1$ is a zero map\footnote{
This is shown around Eq.~(192) of Ref.~\cite{PhysRevResearch.3.043123SM}. 
},
and noting that $\ker S_{11} = \bm 0$, 
we have the following exact sequence, 
\begin{equation}
  \xymatrix{
   0 \ar[r] 
   &
   \ker S
   \ar[r] 
   &
   \ker S'
  \ar[r] 
  & 0
 }  , 
\end{equation}
which implies the isomorphism, 
\begin{equation}
\ker S \simeq \ker S'. 
\end{equation}
Let us explicitly construct the bijection. 
Suppose that we picked an element of $\ker S$, 
\begin{equation}
\bm c = 
\begin{bmatrix}
    \bm c_1 \\
    \bm c_2
\end{bmatrix}
\in \ker S. 
\end{equation}
Since $\bm c \in \ker S$, 
\begin{equation}
S \bm c = 
\begin{bmatrix}
S_{11} \bm c_1 + S_{12} \bm c_2 \\
S_{21} \bm c_1 + S_{22} \bm c_2 
\end{bmatrix}
= \bm 0. 
\label{eq:sc-0}
\end{equation}
We can solve this for $\bm c_1$ as 
(note that $\ker S_{11}$ is trivial) 
$
\bm c_1 = - S^+_{11} S_{12} \bm c_2.$ 
Plugging this into the second equation of Eq.~\eqref{eq:sc-0}, we have 
\begin{equation}
S' \bm c_2 = \bm 0,
\end{equation}
where $S'=S_{22} - S_{21}S_{11}^+ S_{12}$.
Namely, we can get an element of $\ker S'$ from 
$\bm c \in \ker S$ as 
\begin{equation}
\ker S \ni \bm c = 
\begin{bmatrix}
\bm c_1 \\
\bm c_2
\end{bmatrix}
\mapsto 
\bm c_2
\in \ker S'. 
\end{equation}

Conversely,
for a given $\bm c_2 \in \ker S'$, 
we can get an element of $\ker S$ by 
\begin{equation}
\ker  S'
\ni
\bm c_2
\mapsto 
\bm c
\coloneqq  
\begin{bmatrix}
- S_{11}^+ S_{12} \bm c_2
\\
\bm c_2 
\end{bmatrix} 
\in \ker S .
\label{sm-eq:map-sp-s}
\end{equation}
Indeed, 
\begin{equation}
S \bm c
= 
\begin{bmatrix}
(1 - S_{11} S^+_{11}) S_{12} \bm c_2 \\
S' \bm c_2
\end{bmatrix}. 
\label{eq:sc-eq}
\end{equation}
The second line is zero by the assumption that $\bm c_2\in \ker S'$. 
%
%When there is no emergent conserved charge, 
The first line of Eq.~\eqref{eq:sc-eq} also vanishes, which can be shown as follows. 
%
% Since $\widetilde d(\gamma)=0$ by the assumption, 
% for a given $\bm d_1 \in \coker S_{11}$, there always exists an element of $\coker S$, 
% which we denote by 
% $\bm d^\top = (\bm d_1^\top, \bm d_2^\top)$. 
% %
% The condition $\bm d^\top S = \bm 0$ reads 
% \begin{eqnarray}
%   \bm d_2^\top S_{21} &=& \bm 0 , \label{eq:db1}\\
%   \bm d_1^\top S_{12} + \bm d_2^\top S_{22} &=& \bm 0,  \label{eq:db2}
% \end{eqnarray}
% where we used $\bm d^\top_1 S_{11}=\bm 0$.
% %
% The combination 
% $\bm d_1^\top S_{12} \bm c_2$
% can be shown to vanish as follows: 
% %
% \begin{equation}
%   \bm d_1^\top S_{12} \bm c_2 
%   \underset{{\rm (\ref{eq:db2})}}{=}
%   - \bm d_2^\top S_{22} \bm c_2 
%   \underset{S'\bm c_2=\bm 0}{=}
%   - \bm d_2^\top S_{21} S^+_{11} S_{12} \bm c_2 
%   \underset{\rm (\ref{eq:db1})}{=} 
%   0. 
% \end{equation}
% %
% Therefore, we have shown 
% $\bm d_1^\top S_{12} \bm c_2 =  0$ 
% for any $\bm d_1 \in \coker S_{11}$ and $\bm c_2 \in \ker S'$. 
%%%%%% Suggestion
%\color{red}
 By Assumption~\ref{assumption:charge}, for a given $\bm d_1 \in \coker S_{11}, 
 \bm d =
\begin{bmatrix}
\bm d_1 \\
 \bm 0
\end{bmatrix} \in \coker S$. 
 The condition $\bm d^\top S = \bm 0$ can be expressed as follows:
 \begin{eqnarray}
   \bm d_1^\top S_{11} &=& \bm 0 , \label{eq:db1}\\
   \bm d_1^\top S_{12} &=& \bm 0.
   \label{eq:db2}
 \end{eqnarray}
 Therefore, $\bm d_1^\top S_{12} \bm c_2 =  0$ 
 for any $\bm d_1 \in \coker S_{11}$ and $\bm c_2 \in \ker S'$. 
%%%%%%%
This is equivalent to $S_{12} \bm c_2 \in (\coker S_{11})^\perp$. 
Since $(1 - S_{11}S^+_{11})$ is a projection matrix to $\coker S_{11}$, 
the first line of Eq.~\eqref{eq:sc-eq} vanishes. 
Thus, we have $S \bm c=0$ and $\bm c \in \ker S$.

Thanks to the isomorphism~\eqref{sm-eq:map-sp-s}, 
the reaction system obtained by reducing strong buffering structures 
can give the steady-state fluxes of the original system. 
Suppose that we have steady-state fluxes of $\Gamma'$, 
\begin{equation}
S' \bm r_2 = \bm 0. 
\label{eq:sp-r2-0}
\end{equation}
Let us pick a basis of $\ker S'$, 
\begin{equation}
\ker S' = {\rm span} \, \{ \bm c_2^{(\alpha)} \}_{\alpha=1\ldots|\alpha|} .
\end{equation}
The solution of Eq.~\eqref{eq:sp-r2-0} can be 
expanded as 
\begin{equation}
\bm r_2 
=\sum_\alpha \mu^{(\alpha)} (\bm k_2) \, \bm c^{(\alpha)}_2 . 
\end{equation}
From the solution $\bm r_2$, we can construct the 
steady-state fluxes of the whole network, 
\begin{equation}
\bm r = \sum_\alpha \mu^{(\alpha)} (\bm k_2) \, \bm c^{(\alpha)} ,
\end{equation}
where $\bm c^{(\alpha)}$ is defined by 
\begin{equation}
\bm c^{(\alpha)} \coloneqq 
\begin{bmatrix}
- S_{11}^+ S_{12} \bm c_2^{(\alpha)}
\\
\bm c_2^{(\alpha)} 
\end{bmatrix}. 
\end{equation}
The steady-state fluxes constructed this way do not depend on the parameters within the strong buffering structures, $\bm k_1$, because they do not appear 
in the rate equation of $\Gamma'$. 

Although the reconstructed fluxes satisfy the steady-state condition $S \bm r = \bm 0$, 
we should also check that the equations specifying the values of conserved charges 
do not give rise to additional constraints, which could invalidate the solution. 
The conserved charges are specified 
in the system $\Gamma$ by 
\begin{equation}
  \bm d^{(\bar\alpha)} \cdot \bm x 
  = \ell^{\bar\alpha} . \label{eq:dx-l} 
\end{equation}
As shown in Ref.~\cite{PhysRevResearch.3.043123SM}, 
$\coker S$ can be decomposed as 
\begin{equation}
 \coker S \cong  (\coker S)/X(\gamma) \oplus 
 X(\gamma) / \bar D'(\gamma) 
 \oplus \bar D'(\gamma) ,
\end{equation}
where 
\begin{equation}
  \bar D'(\gamma)  \coloneqq 
    \left\{ 
      \begin{pmatrix}
        \bm d_1 \\
        \bm d_2
      \end{pmatrix}
      \in \coker S 
      \, \middle| \, \bm d_1 = \bm 0 
      \right\} .
\end{equation}
Intuitively speaking, 
an element of $(\coker S)/X(\gamma)$ is a conserved charge in $\Gamma$ whose projection to $\gamma$ is not conserved in $\gamma$,
an element of $X(\gamma)  / \bar D' (\gamma)$ is 
a conserved charge of $\Gamma$ whose projection to $\gamma$ is also conserved in $\gamma$, 
and an element of $\bar D'(\gamma)$ is a conserved charge of $\Gamma$ supported in $\Gamma \setminus \gamma$. 
Recall that $d_l (\gamma)$ is written as 
$d_l (\gamma) = |(\coker S) / X(\gamma)|$. 
When $d_l (\gamma) = 0$, $\coker S$ is written as
\begin{equation}
     \coker S 
     \cong D(\gamma) \oplus \bar D' (\gamma),
\end{equation}
where we defined $D(\gamma)\coloneqq X(\gamma) / \bar D'(\gamma)$. 
Correspondingly, we can divide the basis vectors of $\coker S$ 
into two classes, 
$\{ \bm d^{(\bar\alpha)} \} 
= 
\{
   \bm d^{(\bar\alpha_\gamma)}, \bm d^{(\bar \alpha') } 
\}$, 
where 
$\{ \bm d^{(\bar\alpha_\gamma)}\}$ 
is a basis of $D(\gamma)$, 
and 
$\{ \bm d^{(\bar \alpha')} \}$ 
is a basis of $\bar D'(\gamma)$. 
By Assumption~\ref{assumption:charge}, 
the basis vectors are written in the form
\begin{equation}
  \bm d^{(\bar\alpha_\gamma)}
  = 
  \begin{pmatrix}
    \bm d^{(\bar\alpha_\gamma)}_1 
    \\
    \bm 0
  \end{pmatrix},
  \quad 
  \bm d^{(\bar\alpha')}
  = 
  \begin{pmatrix}
    \bm 0 \\
    \bm d^{(\bar\alpha')}_2
  \end{pmatrix}. 
\end{equation}
With this basis of $\coker S$, 
Eq.~(\ref{eq:dx-l}) is written as 
\begin{align}  
 \bm d_1^{(\bar\alpha_\gamma)} \cdot \bm x_1 
  &= \ell^{\bar\alpha_\gamma}, 
\label{eq:d1x1-l}
\\
\bm d_2^{(\bar\alpha')} \cdot \bm x_2 &= \ell^{\bar\alpha'} . 
\label{eq:d2x2-l}
\end{align}
In fact, 
$\bm d^{(\bar\alpha')}_2$ is a conserved charge in $\Gamma'$, 
$\bm d^{(\bar\alpha')}_2 \in \coker S'$, 
as we see below. 
Since $\bm d^{(\bar\alpha')} \in \coker S$ , 
it satisfies 
\begin{equation}
  \begin{pmatrix}
    \bm 0 & (\bm d^{(\bar\alpha')}_2)^\top
  \end{pmatrix}
\begin{pmatrix}
  S_{11} & S_{12} \\
  S_{21} & S_{22}
\end{pmatrix}
= 
\begin{pmatrix}
  (\bm d^{(\bar\alpha')}_2)^\top S_{21} & (\bm d^{(\bar\alpha')}_2)^\top S_{22}
\end{pmatrix}
= \bm 0 . 
\end{equation}
This implies that $\bm d^{(\bar\alpha')}_2$ satisfies 
\begin{equation}
  (\bm d^{(\bar\alpha')}_2)^\top S' 
  = (\bm d^{(\bar\alpha')}_2)^\top ( S_{22} - S_{21} S^{+}_{11} S_{12}) 
  = \bm 0 , 
\end{equation}
hence $\bm d^{(\bar\alpha')}_2 \in \coker S'$. 
Thus we have obtained an injective map, 
\begin{equation}
  \coker S \ni 
  \bm d^{(\bar\alpha')} = 
  \begin{pmatrix}    
    \bm 0 \\
  \bm d^{(\bar\alpha')}_2 
\end{pmatrix}
\mapsto 
\bm d^{(\bar\alpha')}_2  \in \coker S'. 
  \label{eq:sd-ds}
\end{equation}
This map is nothing but the induced map $\bar \varphi_0$\footnote{
The map $\bar \varphi_0$ is written as 
\begin{equation}
\bar\varphi_0: \coker S \ni 
\begin{bmatrix}
    \bm d_1 \\
    \bm d_2
\end{bmatrix}
\mapsto 
[\bm d_2 - S_{21} S^+_{11} S_{12} \bm d_2]
\in \coker S' ,
\end{equation}
where $[\ldots]$ (the second one) means identifying the difference in 
${\rm im\,} S'$ or equivalently projecting to $\coker S'$. 
The space $D(\gamma)$ is isomorphic to the kernel of $\bar\varphi_0$, and the nontrivial image of $\bar \varphi_0$ comes from $\bar D'(\gamma)$ in the current situation where $d_l(\gamma)=0$. 
For an element $
\begin{bmatrix}
    \bm 0 \\
    \bm d_2
\end{bmatrix} 
\in \bar D'(\gamma)$, 
$\bm d_2$ is already in $\coker S'$ as shown above,
so the mapping is written as 
\begin{equation}
\bar \varphi_0: \coker S 
\supset \bar D'(\gamma)
\ni 
\begin{bmatrix}
    \bm 0 \\
    \bm d_2
\end{bmatrix}
\mapsto 
\bm d_2 \in \coker S'. 
\end{equation}
}, 
and when $\widetilde c(\gamma)=0$, this map is a surjection, which follows from the long exact sequence (\ref{eq:les}). 
The equations satisfied by 
the boundary part (denoted by 2) of the concentrations/rates of $\Gamma$ 
are Eqs.~(\ref{eq:sp-r2-0}) and (\ref{eq:d2x2-l}). 
Since all the conserved charges in $\Gamma'$ 
are given as images of $\bar \varphi_0$, 
the equation specifying the conserved charges 
is given by Eq.~\eqref{eq:d2x2-l}. 
Thus, the equations for conserved charges in 
$\Gamma$, Eqs.~\eqref{eq:d1x1-l} and \eqref{eq:d2x2-l} do not give any additional constraints on $\bm x_2$. 
Therefore, the steady-state rates of $\Gamma$ 
can be obtained from those of $\Gamma'$ via the map~\eqref{sm-eq:map-sp-s}. 
This concludes the proof of Theorem~2.

\section{ Details of examples }

In this section, we give the details of the simple example discussed in the main text. 
We also provide a few additional illustrative examples. 

\subsection{ Example 1 }
Here we consider a reaction system $\Gamma$ with four species and six reactions,
\begin{align}
    e_1 &: \text{(input)} \to v_1, \notag \\
    e_2 &: v_1 \to v_2, \notag \\
    e_3 &: v_2 \to v_3 ,\notag \\
    e_4 &: v_3 \to v_4 ,  \\
    e_5 &: v_4 \to v_1 , \notag \\
    e_6 &: v_3 \to \text{(output)} . \notag
\end{align}

The network $\Gamma$ has the following nontrivial buffering structures, 
\begin{align}    
\gamma_1^\ast &= (\{v_1 \}, \{e_2 \}), \\
\gamma_2^\ast &= (\{v_2\},\{ e_3 \} ), \\
\gamma_3^\ast &= (\{v_4\},\{ e_5 \} ),\\
\gamma_4^\ast &= (\{v_1,v_2\},\{e_2,e_3 \} ) = \gamma_1^\ast \cup \gamma_2^\ast, \\
\gamma_5^\ast &= (\{v_1,v_4\},\{e_2,e_5 \} ) = \gamma_1^\ast \cup \gamma_3^\ast, \\
\gamma_6^\ast &= (\{v_2,v_4\},\{e_3,e_5 \} ) = \gamma_2^\ast \cup \gamma_3^\ast, \\
\gamma_7^\ast &= (\{v_1,v_2,v_4\},\{e_2,e_3,e_5 \} ) 
= \gamma_1^\ast \cup \gamma_2^\ast \cup \gamma_3^\ast, \\
\gamma_8 &= (\{v_1,v_2,v_3,v_4\},\{e_2,e_3,e_4,e_5,e_6 \} ) . 
\end{align}
where the ones with $\ast$ are strong buffering structures,
and those without $\ast$ are ordinary buffering structures.
Among the strong buffering structures, $\gamma_7^\ast$ is maximal, and the minimal form is obtained as $\Gamma_{\rm min} = \Gamma / \gamma_7^\ast$. 
Under the reduction, the stoichiometric matrix changes as 
\begin{equation}
\begin{tikzpicture}
\node at (0, 0) {
$
S =
\begin{blockarray}{cccccccc}
&& \\
\begin{block}{c[ccccccc]}
{\color{mydarkred}v_1} \quad\, 
& -1 & 0 & 1 & 1 & 0 & 0 \\
 {\color{mydarkred}v_2} \quad\, 
& 1 & -1 & 0 & 0 & 0 &0 \\
{\color{mydarkred}v_4} \quad\, 
& 0 & 0 & -1 & 0 & 1 &0 \\
v_3 \quad\, 
& 0 & 1 & 0  & 0 & -1 & -1 \\
\end{block} 
& {\color{mydarkred}e_2} 
& {\color{mydarkred}e_3}
& {\color{mydarkred}e_5}
& e_1 
& e_4 
& e_6
\end{blockarray}  
$
}; 
\draw[mydarkred, dashed,line width=1] 
 (-0.9,-0.21) rectangle (0.8,1);
\node at (-0.1, 1.25) {\color{mydarkred}$S_{11}$} ;

\draw [mybrightblue,ultra thick,-latex]   (3,0) -- (4, 0); 

\node at (5.8, 0) {
$
S' =
\begin{blockarray}{ccccc}
&& \\
\begin{block}{c[cccc]}
v_3 \,\, & 1 & 0 & -1 \\
\end{block} 
& e_1 & e_4 & e_6
\end{blockarray}
          $ %\quad .
      }; 

\end{tikzpicture} 
%\label{eq:ex-v4-e6-s-sp}
\end{equation}
where we have brought the components in 
$\gamma_7^\ast$ to the upper-left part. 
The reduction can be visualized as 
\begin{equation}\begin{tikzpicture} 
    \node[species] (v1) at (1.25,0) {$v_1$}; 
    \node[species] (v2) at (2.5,0) {$v_2$};
    \node[species] (v3) at (3.75,0) {$v_3$};
    \node[species] (v4) at (2.5,1.5) {$v_4$}; 
    \node (d1) at (0,0) {}; 
    \node (d2) at (5,0) {}; 
 
    \draw [-latex,draw,line width=0.5mm] (d1) edge node[below]{$e_1$} (v1);
    \draw [-latex,draw,line width=0.5mm] (v1) edge node[below]{$e_2$} (v2);
    \draw [-latex,draw,line width=0.5mm] (v2) edge node[below]{$e_3$} (v3);
    \draw [-latex,draw,line width=0.5mm] (v3) edge node[above right]{$e_4$} (v4);
    \draw [-latex,draw,line width=0.5mm] (v4) edge node[above left]{$e_5$} (v1);
    \draw [-latex,draw,line width=0.5mm] (v3) edge node[below]{$e_6$} (d2);
       
   \draw [mybrightblue,ultra thick,-latex]
   (5.5,0) -- (6.5,0); 
   
    \node[species] (red_v3) at (8.25,0) {$v_3$};
    \node (red_d1) at (7,0) {}; 
    \node (red_d2) at (9.5,0) {}; 

    \draw [-latex,line width=0.5mm] (red_d1) edge node[below]{$e_1$} (red_v3);

    \draw [-latex,line width=0.5mm,loop above,out=60,in=120,looseness=10] 
(red_v3) edge node[above]{$e_4$} (red_v3);

    \draw [-latex,line width=0.5mm] (red_v3) edge node[below]{$e_6$} (red_d2);

\end{tikzpicture}
\end{equation}
A set of basis vectors of $\ker S'$ can be picked as
\begin{equation}
\bm c^{(1)}_2 
= 
\begin{bmatrix}
    1 \\
    0 \\
    1
\end{bmatrix},
\quad 
\bm c^{(2)}_2 
= 
\begin{bmatrix}
    0 \\
    1 \\
    0
\end{bmatrix},
\end{equation}

We can use the rate equation of the reduced system 
to compute the steady-state rates of the original system. 
If we take the mass-action kinetics for example, the rate equation is 
\begin{equation}
\frac{d}{dt} x_3 = k_1 - k_6 x_3. 
\end{equation}
At the steady state, $\bar x_3 = k_1 / k_6$, leading to $\bar r_6=k_6 \bar x_3 =  k_1$ and $\bar r_4 =  k_4 \bar x_3 =k_4 k_1 / k_6$. 
Thus, the steady-state fluxes in the reduced system ($\bar{ \bm r}_2= 
\begin{bmatrix}\bar r_1 & \bar r_4 & \bar r_6
\end{bmatrix}^\top$) can be written as
\begin{equation}
\bar{ \bm r}_2 = 
k_1 \bm c^{(1)}_2 + \frac{k_4 k_1}{ k_6 } \bm c^{(2)}_2 . 
\end{equation}
We can reconstruct $\bm c^{(\alpha)}$ from $\bm c_2^{(\alpha)}$; 
the internal components are 
\begin{align}    
\bm c_1^{(1)}
&= - S^+_{11}S_{12} \bm c_2^{(1)}
= 
- 
\begin{bmatrix}
-1 &0 & -1 \\
-1 & -1 & -1 \\
0 & 0 & -1
\end{bmatrix}
\begin{bmatrix}
1 & 0 & 0 \\
0 & 0 & 0 \\
0 & 1 & 0
\end{bmatrix}
\begin{bmatrix}
1 \\
0 \\
1
\end{bmatrix}
= 
\begin{bmatrix}
1 \\
1 \\
0
\end{bmatrix}, 
\\
\bm c_1^{(2)}
&= - S^+_{11}S_{12} \bm c_2^{(2)}
= 
- 
\begin{bmatrix}
-1 &0 & -1 \\
-1 & -1 & -1 \\
0 & 0 & -1
\end{bmatrix}
\begin{bmatrix}
1 & 0 & 0 \\
0 & 0 & 0 \\
0 & 1 & 0
\end{bmatrix}
\begin{bmatrix}
0 \\
1 \\
0
\end{bmatrix}
= 
\begin{bmatrix}
1 \\
1 \\
1 
\end{bmatrix}, 
\end{align}
where the components are organized in the order $\{e_2,e_3,e_5\}$. 
The steady-state fluxes of the original system read 
\begin{equation}
\bar{ \bm r} = 
k_1 \bm c^{(1)} + \frac{k_4 k_1}{ k_6 } \bm c^{(2)} ,
\end{equation}
where 
\begin{equation}
\bm c^{(1)}
\coloneqq 
\begin{bmatrix}
    \bm c_1^{(1)} \\
    \bm c_2^{(1)} 
\end{bmatrix}, 
\quad 
\bm c^{(2)}
\coloneqq 
\begin{bmatrix}
    \bm c_1^{(2)} \\
    \bm c_2^{(2)} 
\end{bmatrix}. 
\end{equation}

\subsection{ Example 2 } 

We here consider the following reaction network with four species and six reactions: 
\begin{align}
    e_1 &: \text{(input)} \to v_1, \notag \\
    e_2 &: \text{(input)} \to v_2, \notag \\
    e_3 &: v_1 + v_2 \to v_3 ,\\
    e_4 &: v_1 \to v_4 , \notag \\
    e_5 &: v_3 + v_4 \to \text{(output)} . \notag \\
    e_6 &: v_4 \to \text{(output)} . \notag
\end{align}
The network structure is drawn as: 
\begin{equation}
\begin{tikzpicture} 

\node (d1) at (0,9.75) {}; 
\node (d2) at (2,9.75) {}; 

\node[species] (v1) at (0,8.5) {$v_1$}; 
\node[species] (v2) at (2,8.5) {$v_2$};

\node[reaction] (e3) at (1,7.5) {$e_3$};

\node[species] (v3) at (1,6.25) {$v_3$};

\node[species] (v4) at (-1,6.25) {$v_4$};

\node[reaction] (e5) at (0,5) {$e_5$};

\node (d3) at (0,3.75) {}; 

\node (d4) at (-2.5,6.25) {}; 

\draw [-latex,draw,line width=0.5mm] (d1) edge node[left]{$e_1$} (v1);
\draw [-latex,draw,line width=0.5mm] (d2) edge node[left]{$e_2$} (v2);

\draw [-latex,draw,line width=0.5mm] (v1) edge (e3);
\draw [-latex,draw,line width=0.5mm] (v2) edge (e3);

\draw [-latex,draw,line width=0.5mm] (e3) edge (v3);
\draw [-latex,draw,line width=0.5mm] (v1) edge node[left]{$e_4$} (v4);

\draw [-latex,draw,line width=0.5mm] (v4) edge (e5);
\draw [-latex,draw,line width=0.5mm] (v3) edge (e5);
\draw [-latex,draw,line width=0.5mm] (e5) edge (d3);

\draw [-latex,draw,line width=0.5mm] (v4) edge  node[above]{$e_6$} (d4);

\end{tikzpicture}
\end{equation}

The stoichiometric matrix is 
\begin{equation}\label{eq:ex1-stoi-mat}
    S = \begin{bmatrix}
        1 & 0 & -1 & -1 & 0 & 0 \\
        0 & 1 & -1 & 0 & 0 & 0 \\
        0 & 0 & 1 & 0 & -1 & 0 \\
        0 & 0 & 0 & 1 & -1 & -1 
    \end{bmatrix},
\end{equation}
and reaction rate functions are  
\begin{equation}
\begin{split}
\bm r (\bm x)    
&= 
\begin{bmatrix}
r_1 & r_2 & r_3 (x_1,x_2) & r_4 (x_1) & r_5 (x_3,x_4) & r_6 (x_4)
\end{bmatrix}^\top 
\\
&= 
\begin{bmatrix}
k_1 & k_2 & k_3 x_1x_2 & k_4x_1 & k_5 x_3 x_4 & k_6x_4 
\end{bmatrix}^\top . 
\end{split}
\end{equation}

We can identify all the nontrivial buffering structures in this example as 
\begin{align}
\gamma^\ast_1 &= (\{v_2\},\{ e_3 \}), \\
\gamma^\ast_2 &= (\{v_3\},\{ e_5 \}), \\
\gamma^\ast_3 &= (\{v_1,v_2\},\{ e_3,e_4 \}), \\
\gamma^\ast_4 &= (\{v_2,v_3\},\{ e_3,e_5 \}) = \gamma^\ast_1 \cup \gamma^\ast_2 , \\
\gamma^\ast_5 &= (\{v_3,v_4\},\{ e_5,e_6 \}), \\
\gamma^\ast_6 &= (\{v_1,v_2,v_3\},\{ e_3,e_4,e_5 \}) = \gamma^\ast_2 \cup \gamma^\ast_3, \\
\gamma^\ast_7 &= (\{v_2,v_3,v_4\},\{ e_3,e_5,e_6 \}), \\
\gamma^\ast_8 &= (\{v_1,v_2,v_3,v_4\},\{ e_3,e_4,e_5,e_6 \}) = \gamma^\ast_3 \cup \gamma^\ast_5 . 
\end{align}
All the eight buffering structures are strong buffering structures. Among these, $\gamma^\ast_8$ is maximal, i.e., all the strong buffering structures are subsets of $\gamma^\ast_8$. Thus, by eliminating $\gamma^\ast_8$, we can get the minimal form of the given reaction network, $\Gamma' \coloneqq \Gamma / \gamma^\ast_8$  with no species and two reactions $e_1$ and $e_2$. 
The stoichiometric matrix of the minimal form, $S'$, is given as follows:
\begin{equation}
\begin{tikzpicture}
\node at (0, 0) {
$
S =
\begin{blockarray}{cccccccc}
&& \\
\begin{block}{c[ccccccc]}
{\color{mydarkred}v_1} \quad\, 
& -1 & -1 & 0 & 0 & 1 & 0 \\
 {\color{mydarkred}v_2} \quad\, 
& -1 & 0 & 0 & 0 & 0 & 1 \\
{\color{mydarkred}v_3} \quad\, 
& 1 & 0 & -1 & 0 & 0 & 0 \\
{\color{mydarkred}v_4} \quad\, 
& 0 & 1 & -1 & -1 & 0 & 0 \\
\end{block} 
& {\color{mydarkred}e_3} & {\color{mydarkred}e_4}
& {\color{mydarkred}e_5} & {\color{mydarkred}e_6} & e_1 & e_2
\end{blockarray}  
$
}; 
\draw[mydarkred, dashed,line width=1] 
 (-1.0,-0.7) rectangle (1.5,1.1);
\node at (0.2, 1.4) {\color{mydarkred}$S_{11}$} ;

\draw [mybrightblue,ultra thick,-latex]   (3,0) -- (4, 0); 

\node at (5.5, 0) {
$
S' =
\begin{blockarray}{cccc}
&& \\
&& \\
\begin{block}{c[ccc]}
\,\, & \cdot & \cdot \\
\end{block} 
& e_1 & e_2
\end{blockarray}
          $ %\quad .
      }; 

\end{tikzpicture} 
\label{eq:ex-1-s-sp}
\end{equation}
Note that the reduced stoichiometric matrix, $S'$, of size $0 \times 2$, which means that it is a linear map from the two-dimensional vector space to the zero-dimensional one. 
In the minimal form, the steady-state fluxes
is obtained trivially as 
\begin{equation}
\begin{bmatrix}
r_1 \\
r_2
\end{bmatrix}
= 
\begin{bmatrix}
k_1 \\
k_2 
\end{bmatrix}
= 
k_1 \, \bm c_2^{(1)}
+ 
k_2 \, \bm c_2^{(2)},
\end{equation}
where we picked basis vectors of $\ker S'$ as 
$\bm c_2^{(1)} \coloneqq
\begin{bmatrix} 
1 & 0 
\end{bmatrix}^\top
$
and 
$
\bm c_2^{(2)} \coloneqq
\begin{bmatrix}
0 & 1
\end{bmatrix}^\top
$. 
Using this, we can compute the steady-state fluxes of the original system as 
\begin{equation}\label{eq:rate-ss}
\bar{\bm r} 
    = 
    k_1 %\mu^{(1)}(\bm k)
    \bm c^{(1)}
    + 
    k_2 %\mu^{(2)}(\bm k)
    \bm c^{(2)}, 
\end{equation}
where 
$\bm c^{(1)} \coloneqq 
\begin{bmatrix}
    1 & 0 & 0 & 1 & 0 & 1
\end{bmatrix}^\top
$
and 
$
\bm c^{(2)} \coloneqq 
\begin{bmatrix}
    0 & 1 & 1 & -1 & 1 & -2
\end{bmatrix}^\top
$ are the images of 
$\bm c_2^{(1)}$ and $\bm c_2^{(2)}$ 
via the map $\ker S' \to \ker S$, 
and they form a basis of $\ker S$. 
Since the rate constants of the reactions contained in the strong buffering structure do not affect the steady-state reaction fluxes, $\bar{\bm r}$ is independent of the rate constants $k_3, k_4, k_5$, and $k_6$ as we can easily check in Eq.~\eqref{eq:rate-ss}.

\subsection{ Example 3 } 

Here we describe an example with a conserved charge. We consider a network $\Gamma = (\{v_1,v_2,v_3,v_4 \},\{e_1,e_2,e_3,e_4,e_5 \})$ 
with the following reactions, 
\begin{align}
    e_1 &: \text{(input)} \to v_1, \notag \\
    e_2 &: v_1 \to v_2, \notag \\
    e_3 &: v_2 \to \text{(output)} ,\\
    e_4 &: v_1+v_2 \to v_3+v_4 , \notag \\
    e_5 &: v_3+v_4 \to v_1+v_2. \notag
\end{align}

The network structure can be drawn as 
\begin{equation}
    \begin{tikzpicture}[bend angle=45] 
    \node[species] (v1) at (0,0) {$v_1$}; 
    \node[species] (v2) at (2,0) {$v_2$};
 
    \node (d1) at (-1.25,0) {};
    \node (d2) at (3.25,0) {};
    
    \node[reaction] (e4) at (0, -1.65) {$e_4$}; 
    \node[reaction] (e5) at (2,-1.65) {$e_5$}; 

    \node[species] (v3) at (0,-3.25) {$v_3$}; 
    \node[species] (v4) at (2,-3.25) {$v_4$};

    \draw[-latex, line width=0.5mm] (d1) edge node[below] {$e_1$} (v1); 
    \draw[-latex, line width=0.5mm] (v1) edge node[below] {$e_2$} (v2); 
    \draw[-latex, line width=0.5mm] (v2) edge node[below] {$e_3$} (d2); 
  
    \draw[-latex, line width=0.5mm] (v1) -- (e4); 
    \draw[-latex, line width=0.5mm] (v2) -- (e4); 
    \draw[-latex, line width=0.5mm] (e4) -- (v3); 
    \draw[-latex, line width=0.5mm] (e4) -- (v4); 
 
    \draw[-latex, line width=0.5mm] (v3) -- (e5); 
    \draw[-latex, line width=0.5mm] (v4) -- (e5); 
    \draw[-latex, line width=0.5mm] (e5) -- (v1); 
    \draw[-latex, line width=0.5mm] (e5) -- (v2); 
 
\end{tikzpicture}
\end{equation}
If we use the mass-action kinetics, 
the rate equation is written as 
\begin{equation}
    \frac{d}{dt}
    \begin{bmatrix}
    x_1 \\
    x_2 \\
    x_3 \\
    x_4 
    \end{bmatrix}
    = 
    \begin{bmatrix}
      1 & -1 & 0 & -1 & 1 \\
      0 & 1 & -1 & -1 & 1 \\
      0 & 0 & 0 & 1 & -1 \\
     0 & 0 & 0 & 1 & -1 
    \end{bmatrix}
    \begin{bmatrix}
      r_1 \\
      r_2 \\
      r_3 \\
      r_4 \\
      r_5 
    \end{bmatrix}, 
    \quad 
    \begin{bmatrix}
      r_1 \\
      r_2 \\
      r_3 \\
      r_4 \\
      r_5 
    \end{bmatrix}
    = 
    \begin{bmatrix}
      k_1 \\
      k_2 x_1 \\
      k_3 x_2 \\
      k_4 x_1 x_2 \\
      k_5 x_3 x_4 
    \end{bmatrix}.
\end{equation}
The kernel and cokernel of the stoichiometric matrix are given by 
\begin{align}
    \ker S &= {\rm span\,} \{ 
    \begin{bmatrix}
      1 & 1 & 1 & 0 & 0 
    \end{bmatrix}^\top,
    \begin{bmatrix}
      0 & 0 & 0 & 1 & 1 
    \end{bmatrix}^\top
    \},  \\
\coker S &= {\rm span\,} \{ 
    \begin{bmatrix}
     0 & 0 & 1 & -1 
    \end{bmatrix}^\top 
    \} .
\end{align}
The cokernel is one-dimensional and 
the system has one conserved charge. 
To find the steady states, we need to specify the value 
of the charge as 
\begin{equation}
    \ell = x_3 - x_4 . 
\end{equation}
The steady-state reaction rates and concentrations are 
\begin{align}
    \bar{\bm r}
    &= k_1 
    \begin{bmatrix}
      1 & 1 & 1 & 0 & 0 
    \end{bmatrix}^\top
    +
    \frac{k_4 k_1^2}{k_2 k_3} 
    \begin{bmatrix}
      0 & 0 & 0 & 1 & 1 
    \end{bmatrix}^\top ,\\ 
    \bar {\bm x} 
    &= 
    \begin{bmatrix}
      \frac{k_1}{k_2} & \frac{k_1}{ k_3}
      & \frac{1}{2} 
      \left(
       \ell + \sqrt{\ell^2 + 4 K}
      \right)
      & 
      & \frac{1}{2} 
      \left(
      - \ell + \sqrt{\ell^2 + 4 K}
      \right)
    \end{bmatrix}^\top, 
\end{align}
where we set $K \coloneqq k_4 k_1^2 / k_2 k_3$.
The buffering structures (except for the network itself) in this example are
\begin{align}
\gamma^\ast_1 &= (\{v_3,v_4 \},\{e_5 \} ), \\
\gamma_2 &= (\{v_1,v_3,v_4 \},\{e_2,e_4,e_5 \} ), \\
\gamma_3 &= (\{v_2,v_3,v_4 \},\{e_3,e_4,e_5 \} ) .
\end{align}
Among these, $\gamma_1^\ast$ is a strong buffering structure, which includes a conserved charge $\ell = x_3 - x_4$. Correspondingly, the steady-state reaction rate does not depend on $k_5$ and $\ell$. 

The maximal strong buffering structure is $\gamma_1^\ast$. 
Under the reduction to the minimal form, $\Gamma_{\rm min}\coloneqq \Gamma / \gamma_1^\ast$, the stoichiometric matrix changes as
\begin{equation}
\begin{tikzpicture}
\node at (0, 0) {
$
S =
\begin{blockarray}{ccccccc}
&& \\
\begin{block}{c[cccccc]}
{\color{mydarkred}v_3} \quad\, 
& -1& 0 & 0 & 0 & 1  \\
 {\color{mydarkred}v_4} \quad\, 
&-1 & 0 & 0 & 0 & 1  \\
v_1 \quad\, 
& 1 & 1 & -1 & 0 & -1  \\
v_2 \quad\, 
& 1 & 0 & 1 & -1 & -1  \\
\end{block} 
& {\color{mydarkred}e_5} 
& e_1
& e_2
& e_3
& e_4 
\end{blockarray}  
$
}; 
\draw[mydarkred, dashed,line width=1] 
 (-0.9,0.2) rectangle (0,1);
\node at (-0.5, 1.25) {\color{mydarkred}$S_{11}$} ;

\draw [mybrightblue,ultra thick,-latex]   (3,0) -- (4, 0); 

\node at (6.3, 0) {
$
S' =
\begin{blockarray}{cccccc}
&& \\
\begin{block}{c[ccccc]}
v_1 \,\, & 1 & -1 & 0 & 0 \\
v_2 \,\, & 0 & 1 & -1 & 0 \\
\end{block} 
& e_1 & e_2 & e_3 & e_4
\end{blockarray}
          $ %\quad .
      }; 

\end{tikzpicture} 
\label{eq:ex-v4-e5-s-sp}
\end{equation}
where we have brought the components in 
$\gamma_1^\ast$ to the upper-left part. 

\section{Metabolic pathways of {\it E.~coli}}

Here we describe the details of the metabolic pathways of {\it E.~coli} discussed in the main text. 

\subsection{List of reactions}\label{sec:ecoli-list}

1: Glucose  +  PEP  $\rightarrow$  G6P  +  PYR. 
 
 2: G6P   $\rightarrow$  F6P. 
 
 3: F6P   $\rightarrow$  G6P.
  
  4: F6P  $\rightarrow$  F16P. 
   
   5: F16P  $\rightarrow$  G3P  +  DHAP. 
   
   6: DHAP   $\rightarrow$  G3P.
   
 7: G3P   $\rightarrow$  3PG. 
  
  8: 3PG   $\rightarrow$  PEP. 
  
  9: PEP   $\rightarrow$  3PG.  
  
  10: PEP   $\rightarrow$  PYR. 
  
  11: PYR   $\rightarrow$  PEP. 
  
  12: PYR    $\rightarrow$  AcCoA  +   CO2. 
  
  13: G6P  $\rightarrow$  6PG. 
  
  14: 6PG  $\rightarrow$   Ru5P  +  CO2. 
  
  15: Ru5P  $\rightarrow$  X5P.
  
   16: Ru5P  $\rightarrow$   R5P. 
   
   17: X5P  +  R5P  $\rightarrow$   G3P  +  S7P. 
   
   18: G3P  +  S7P  $\rightarrow$   X5P  +  R5P. 
   
   19: G3P  +  S7P  $\rightarrow$   F6P  +  E4P.
   
    20: F6P  +  E4P  $\rightarrow$  G3P  +  S7P. 
    
21:  X5P  +  E4P  $\rightarrow$   F6P  +  G3P. 
  
  22: F6P   +  G3P  $\rightarrow$   X5P  +  E4P.  
  
  23: AcCoA  +  OAA  $\rightarrow$  CIT. 
  
  24: CIT   $\rightarrow$  ICT. 
  
  25: ICT  $\rightarrow$  2${\rm \mathchar`-}$KG  +  CO2. 
  
  26: 2-KG  $\rightarrow$   SUC  +  CO2.
  
   27: SUC   $\rightarrow$  FUM. 
   
  28:  FUM  $\rightarrow$  MAL. 
   
   29: MAL   $\rightarrow$  OAA.
   
 30: OAA   $\rightarrow$  MAL.

 31: PEP  +  CO2  $\rightarrow$  OAA.

 32: OAA  $\rightarrow$  PEP  +   CO2. 

 33: MAL  $\rightarrow$   PYR  +  CO2.

34: ICT   $\rightarrow$  SUC  +  Glyoxylate. 

 35: Glyoxylate  +  AcCoA  $\rightarrow$  MAL. 

 36: 6PG  $\rightarrow$   G3P  +  PYR. 

 37: AcCoA  $\rightarrow$   Acetate. 

38:  PYR  $\rightarrow$  Lactate. 

 39: AcCoA  $\rightarrow$  Ethanol. 

 40: R5P  $\rightarrow$ (output).

 41: OAA  $\rightarrow$ (output).

 42: CO2  $\rightarrow$ (output).

43:  (input) $\rightarrow$  Glucose. 

 44:  Acetate $\rightarrow$ (output).
 
  45: Lactate $\rightarrow$ (output).

46:  Ethanol $\rightarrow$ (output).
\\

\subsection{ Buffering structures }\label{sec:ecoli-buffering-list}

We here give the list of 
buffering structures 
in the metabolic pathways of {\it E.~coli}. 
Strong buffering structures 
are indicated by $\ast$.

$\gamma^\ast_1=(\{ \rm 
Glucose
\},\{
1
\})$,

$\gamma_2=(\{ \rm 
Glucose,
PEP,
G6P,
F6P,
F16P,
DHAP,
G3P,
3PG,
PYR,
6PG,
Ru5P,
X5P,
R5P, 
S7P,
E4P,
AcCoA,
OAA, \\
CIT,
ICT, 
2{\rm \mathchar`-KG},
SUC,
FUM,
MAL,
CO2,
Glyoxylate,
Acetate,
Lactate,
Ethanol
\},  \{
1,
2 ,
3 ,
4 ,
5 ,
6 ,
7 ,
8 ,
9 ,
10 ,
11 ,
12 ,
13 ,
14 , \\
15 , 
16 ,
17 ,
18 ,
19 ,
20 , 
21 ,
22 ,
23 ,
24 ,
25 ,
26 ,
27 , 
28 ,
29 ,
30 ,
31 ,
32 ,
33 ,
34 ,
35 ,
36 ,
37 ,
38 ,
39 ,
40 ,
41 ,
42 ,
44 ,
45 ,
46 
\})$,

$\gamma^\ast_3=(\{ \rm 
F16P
\},\{
5 
\})$,

$\gamma^\ast_4=(\{ \rm 
DHAP
\},\{
6 
\})$,

$\gamma_5=(\{ \rm 
G3P,
X5P,
S7P,
E4P
\},\{
7 ,
17 ,
18 ,
19 ,
20 ,
21 ,
22 
\})$,

$\gamma^\ast_6=(\{ \rm 
3PG
\},\{
8 
\})$,

$\gamma_7=(\{ \rm 
Glucose,
PEP,
3PG,
PYR,
AcCoA,
OAA,
CIT,
ICT,
2{\rm \mathchar`-KG},
SUC,
FUM, 
MAL,
CO2,
Glyoxylate,
Acetate, \\
Lactate,  
Ethanol
\},
\{
1,
8 ,
9 ,
10 ,
11 ,
12 ,
23 ,
24 ,
25 ,
26 ,
27 ,
28 ,
29 ,
30 ,
31 ,
32 ,
33 ,
34 ,
35 ,
37 ,
38 ,
39 ,
41 ,
42 ,
44 ,
45 ,
46 
\})$,

$\gamma_8=(\{ \rm 
X5P,
S7P,
E4P
\},\{
17 ,
18 ,
19 ,
20 ,
21 
\})$ ,

$\gamma^\ast_9=(\{ \rm 
CIT
\},\{
24 
\})$,

$\gamma^\ast_{10}=(\{ \rm 
2{\rm \mathchar`-KG}
\},\{
26 
\})$,

$\gamma^\ast_{11}=(\{ \rm 
SUC
\},\{
27 
\})$ ,

$\gamma^\ast_{12}=(\{ \rm 
FUM
\},\{
28 
\})$ 

$\gamma^\ast_{13}=(\{ \rm 
Glyoxylate
\},\{
35 
\})$,

$\gamma_{14}=(\{ \rm 
X5P,
R5P,
S7P,
E4P
\},\{
17 ,
18 ,
19 ,
20 ,
21 ,
40 
\})$ ,

$\gamma^\ast_{15}=(\{ \rm 
Acetate
\},\{
44 
\})$,

$\gamma^\ast_{16}=(\{ \rm 
Lactate
\},\{
45 
\})$,

$\gamma^\ast_{17}=(\{ \rm 
Ethanol
\},\{
46 
\})$.
\\

The maximal strong buffering structure is $\bar\gamma_{\rm s} = \gamma^\ast_{1} \cup \gamma^\ast_{3} \cup \gamma^\ast_{4} \cup \gamma^\ast_{6} \cup \gamma^\ast_{9} \cup \gamma^\ast_{10} \cup \gamma^\ast_{11} \cup \gamma^\ast_{12} \cup \gamma^\ast_{13} \cup \gamma^\ast_{15} \cup \gamma^\ast_{16} \cup \gamma^\ast_{17}. $

\subsection{ Parameters for the simulation in Figure~2 }\label{sec:ecoli-parameter-list}

In the simulation shown in Fig.~2, we employ the mass-action kinetics. In the simulation, the initial concentrations and the reaction rate constants are chosen as: 
$
x_{\rm AcCoA} = 0.4,
x_{\rm Acetate} = 1.0,
x_{\rm CIT} = 0.2,
x_{\rm CO2} = 0.6,
x_{\rm DHAP} = 0.2,
x_{\rm E4P} = 0.2,
x_{\rm Ethanol} = 0.5,
x_{\rm F16P} = 0.5,
x_{\rm F6P} = 0.3,
x_{\rm FUM} = 0.1,
x_{\rm G3P} = 1.3,
x_{\rm G6P} = 0.9,
x_{\rm Glucose} = 3.1,
x_{\rm Glyoxylate} = 0.1,
x_{\rm ICT} = 0.1,
x_{\rm KG2} = 0.1,
x_{\rm Lactate} = 0.5,
x_{\rm MAL} = 0.1,
x_{\rm OAA} = 0.1,
x_{\rm PEP} = 0.3,
x_{\rm PG3} = 1.2,
x_{\rm PG6} = 0.3,
x_{\rm PYR} = 0.8,
x_{\rm R5P} = 0.1,
x_{\rm Ru5P} = 0.1,
x_{\rm S7P} = 0.1,
x_{\rm SUC} = 0.1,
x_{\rm X5P} = 4.0,
\text{ and }
k_1 = 5.0,
k_2 = 4.7,
k_3 = 7.8,
k_4 = 5.7,
k_5 = 3.8,
k_6 = 9.7,
k_7 = 5.0,
k_8 = 6.2,
k_9 = 3.5,
k_{10} = 9.8,
k_{11} = 2.5,
k_{12} = 6.1,
k_{13} = 4.0,
k_{14} = 3.8,
k_{15} = 7.8,
k_{16} = 6.6,
k_{17} = 3.8,
k_{18} = 5.5,
k_{19} = 5.7,
k_{20} = 4.7,
k_{21} = 8.0,
k_{22} = 7.3,
k_{23} = 9.2,
k_{24} = 1.1,
k_{25} = 9.6,
k_{26} = 7.4,
k_{27} = 7.4,
k_{28} = 8.3,
k_{29} = 6.2,
k_{30} = 6.4,
k_{31} = 6.2,
k_{32} = 7.9,
k_{33} = 9.1,
k_{34} = 6.7,
k_{35} = 1.6,
k_{36} = 9.6,
k_{37} = 4.7,
k_{38} = 5.1,
k_{39} = 7.3,
k_{40} = 3.8,
k_{41} = 8.4,
k_{42} = 9.7,
k_{43} = 4.8,
k_{44} = 2.0,
k_{45} = 8.0,
k_{46} = 3.7.
$
At time $t = 10$, the rate constant $k_8$, which is inside the maximum strong buffering structure $\gamma^\ast$, is tripled. At time $t=30$ the rate constant $k_2$, which is outside $\gamma^\ast$, is doubled.%\appendix

%\appendix

%\bibliographystyle{utphys.bst}

%\bibliography{refs}

\end{document}